\documentclass{article}

\usepackage[utf8]{inputenc}
\usepackage[12pt]{extsizes}
\usepackage{graphicx}
\usepackage{hyperref}
\usepackage{amsmath}
\usepackage{amssymb}
\usepackage{tabularx, booktabs}
\usepackage{caption}
\usepackage[margin=1in]{geometry}
\usepackage[
    style=apa,
    backend=biber,
    uniquelist=false,
]{biblatex}
\usepackage{placeins}

\newcommand{\expect}[1]{\mathbb{E}\left[#1\right]}
\newcommand{\var}[1]{\mathbb{V}\left[#1\right]}
\newcommand{\vect}[1]{\text{vec}\left(#1\right)}

\newcommand{\norm}[1]{\left\lVert #1 \right\rVert}

\usepackage{todonotes}
\usepackage{tikz}
\usetikzlibrary{matrix}
\usetikzlibrary{shapes.geometric}
\usepackage{varwidth}
\usepackage[automake]{glossaries}
\usepackage{listings}
\usepackage{xcolor}
\usepackage{datetime}
\usepackage{multirow}
\usepackage{amsthm}
\usepackage{authblk}
\newdateformat{monthyear}{\monthname[\THEMONTH] \THEYEAR}

\newtheorem{proposition}{Proposition}
\newtheorem{corollary}{Corollary}[section]

\makeglossaries

\newacronym{ols}{OLS}{ordinary least squares}
\newacronym{wls}{WLS}{weighted least squares}
\newacronym{mint}{MinT}{minimum trace}
\newacronym{rls}{RLS}{recursive least squares}
\newacronym{rmse}{RMSE}{root mean squared error}
\newacronym{rrmse}{RRMSE}{relative reduction in root mean squared error}
\newacronym{glm}{LM}{linear model}
\newacronym{gls}{GLS}{generalized least squares}
\newacronym{map}{MAP}{maximum a posteriori}
\newacronym{mse}{MSE}{mean squared error}
\newacronym{spd}{SPD}{symmetric positive definite}
\newacronym{mle}{MLE}{maximum likelihood estimate}
\newacronym{dh}{DH}{district heating}

\usetikzlibrary{arrows, positioning}

\title{Online forecast reconciliation using linear models}
\author[1]{Tobias Rønlev-Knudsen\thanks{Corresponding author: tobro@dtu.dk}}
\author[1]{Henrik Madsen}
\author[1]{Jan Kloppenborg Møller}  
\affil[1]{Department of Applied Mathematics and Computer Science, Technical University of Denmark}
\date{\monthyear\today}

\begin{document}
\maketitle

\begin{abstract}
    We present a framework for online and adaptive forecasting and hierarchical reconciliation using linear regression models. We begin by formalizing hierarchies using graphs, and motivated by their structure, formulate a multivariate \glsentrylong{glm} using the matrix normal distribution to characterize residuals. Parameter estimation is posed as a ridge regression problem and applied to hierarchical forecast reconciliation. The connections between ridge regression, Bayesian estimation and shrinkage for hierarchical reconciliation are discussed, and results for uncertainty quantification in parameters and forecasts are provided. Based on the ridge regression formulation, a recursive inference scheme inspired by \glsentrylong{rls} is described. The algorithm is implemented in the \textit{PyOnlineForecast} package. Finally, the proposed methodology is demonstrated on a case study for \glsentrylong{dh} load forecasting using a temporal hierarchy. Our results provide a reference for implementation of forecast reconciliation via multivariate linear models in an online setting. The case study furthermore highlights practical considerations of using temporal hierarchies in an online setting and demonstrates the usefulness of the proposed framework and implementation, both for \glsentrylong{dh} load forecasting and more generally for online hierarchical forecasting.
\end{abstract}
\section{Introduction}\label{sec:introduction}
% Forecast reconciliation
Forecast reconciliation is a method to ensure coherency of forecasts of time series with hierarchical structures. Such structures enforce linear constraints on the forecasts, which are exploited by reconciliation methods to improve accuracy. 
% Use cases
Forecast reconciliation has successfully been applied to a range of domains. \cite{athanasopoulosHierarchicalForecastsAustralian2009} considered domestic tourism in Australia, grouping high dimensional time series data by geographical regions and purpose of travel, and demonstrated the use of several methods for achieving coherent forecasts. Applications within energy systems have also been considered. Examples include photovoltaic \parencite{yagliReconcilingSolarForecasts2019}, wind power \parencite{baiDistributedReconciliationDayAhead2019}, electricity loads \parencite{nystrupTemporalHierarchiesAutocorrelation2020,nystrupDimensionalityReductionForecasting2021} and \gls{dh} load forecasting \parencite{bergsteinssonHeatLoadForecasting2021,bergsteinssonHeatLoadForecasting2023}.
% History/related work
Early works on hierarchical time series used bottom-up or top-down approaches to aggregate or disaggregate forecasts and ensure coherency \parencite{grossDisaggregationMethodsExpedite1990,athanasopoulosHierarchicalForecastsAustralian2009}. As the basis for many later works, \cite{hyndmanOptimalCombinationForecasts2011} introduced the optimal combination method, which expresses coherent forecasts as a linear combination of base forecasts by solving an \gls{ols} problem. For large hierarchies \parencite{hyndmanFastComputationReconciled2016} considered a \gls{wls} approach utilizing sparse matrix solutions to improve efficiency. In general, these methods consider a \gls{gls} problem that uses the covariance matrix of the so-called coherency errors. This covariance matrix is usually unknown, and it was shown by \cite{wickramasuriyaOptimalForecastReconciliation2019} to be unidentifiable. In place of the true covariance matrix, they provided theoretical justification for using the full variance-covariance matrix of the base forecast errors using the \gls{mint} method. They furthermore provided several suggestions for estimators of the covariance matrix, including shrinkage based estimators \parencite{schaferShrinkageApproachLargeScale2005} to address collinearity, which is common in hierarchical time series. \cite{pritulargaStochasticCoherencyForecast2021} proposed that coherence should be considered stochastic, and considered the sources and effect of uncertainty in forecast reconciliation.

For temporal hierarchies, \cite{mollerLikelihoodbasedInferenceTemporal2024} studied the importance of correct specification of the covariance matrix, and suggested dimensionality reduction by parameterization of the covariance matrix to improve estimation. Further work by \cite{mollerOptimalForecastReconciliation2024} also investigated uncertainty quantification in hierarchies and showed that forecast reconciliation can be stated as a \gls{glm}. This result allows the use of methods for \gls{glm}s, and enabled explicit quantification of the uncertainties in the weight matrices used for reconciliation. The study demonstrated that forecast reconciliation acts as a linear projection, mapping base forecasts to a coherent linear subspace. This was also noted by \cite{panagiotelisForecastReconciliationGeometric2021}, who provided a geometric interpretation of forecast reconciliation.

% Online and adaptive methods
In practical settings, forecasts are often required to be updated regularly as new observations become available. Online forecasting methods address this by providing efficient algorithms for updating model parameters, typically using recursive estimation techniques. Adaptive models further ensure that model parameters are updated to reflect, typically slow, changes in the modelled system. 

% District heating load forecasting
Forecasting of \gls{dh} loads using adaptive methods was considered by \cite{bacherShorttermHeatLoad2013}. They proposed linear time series models using climate inputs and Fourier series of calendar variables as input features for predicting heat loads. Adaptive methods have also been considered for hierarchical models for heat load forecasting \parencite{bergsteinssonHeatLoadForecasting2021,bergsteinssonHeatLoadForecasting2023}. These works have considered both temporal and spatial hierarchies and demonstrate the use of \gls{mint} reconciliation using shrinkage based on \cite{ledoitImprovedEstimationCovariance2003} in an online and adaptive setting.

% Structure of the paper
The main contribution of this paper is to apply the results of \cite{mollerOptimalForecastReconciliation2024} to derive online and adaptive forecasting methods for hierarchical time series. Moreover, the paper extends current descriptions of forecast reconciliation by using graphs to define hierarchies and matrix normal distributions to characterize errors.
The paper is accompanied by an open source implementation in the Python package \textit{PyOnlineForecast} \parencite{Ronlev-Knudsen_PyOnlineForecast_2026}. The package is loosely based on the \textit{onlineforecast} package \parencite{bacherOnlineforecastPackageAdaptive2023} for the R programming language, which provides methods for online forecasting using ensembles of linear models and the \gls{rls} algorithm for adaptive parameter estimation. As a demonstration of the methods, a \gls{dh} load forecasting problem is considered. 

The structure of the paper is as follows: In Section \ref{sec:hierarchical_reconciliation}, hierarchies and forecast reconciliation are formally introduced. Sections \ref{sec:GLM_regularization} and \ref{sec:recursive_estimation} then describe a regularized multivariate linear model and results on recursive inference for this model. Section \ref{sec:GLM_hierarchies} demonstrates how linear forecast reconciliation may be interpreted as a special case of the model. In Section \ref{sec:implementation} a brief discussion of the implementation in the \textit{PyOnlineForecast} package is given. Section \ref{sec:adaptive_heat_load_forecasting} provides an example demonstrating the use of the methods and package on a \gls{dh} load forecasting problem. Section \ref{sec:conclusion} finally concludes the paper with a summary and discussion of results.
\section{Hierarchical reconciliation}\label{sec:hierarchical_reconciliation}
% Briefly, what is hierarchical reconciliation?
Hierarchical forecast reconciliation uses the notion of a hierarchy to define a linear summation constraint on the states of a system. The summation constraints are commonly illustrated by drawing the hierarchy as in Figure \ref{fig:hierarchy_example}. 
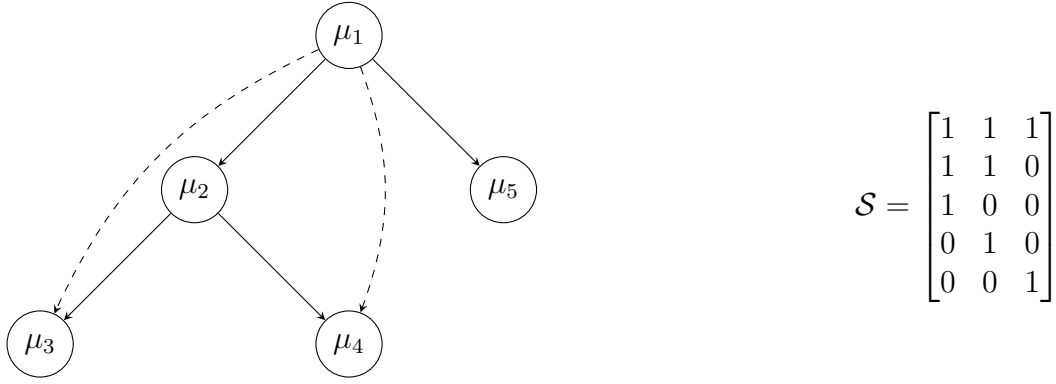
\begin{figure}
    \centering
    \begin{minipage}{0.45\textwidth}
        \centering
        \begin{tikzpicture}[->, >=stealth, node distance=2cm, every node/.style={draw, circle}]
            \node (u1) {$\mu_1$};
            \node (u2) [below left=of u1] {$\mu_2$};
            \node (u3) [below left=of u2] {$\mu_3$};
            \node (u4) [below right=of u2] {$\mu_4$};
            \node (u5) [below right=of u1] {$\mu_5$};

            % Main directed edges
            \draw (u1) -- (u2);
            \draw (u1) -- (u5);
            \draw (u2) -- (u3);
            \draw (u2) -- (u4);

            % Curvy dashed edges for transitive closure
            \draw[dashed, bend right=20] (u1) to (u3);
            \draw[dashed, bend left=20] (u1) to (u4);
        \end{tikzpicture}
    \end{minipage}
    \hfill
    \begin{minipage}{0.45\textwidth}
        \centering
        % Replace the matrix below with your desired expression.
        \[
        \mathcal{S} = \begin{bmatrix}
            1 & 1 & 1 \\
            1 & 1 & 0 \\
            1 & 0 & 0 \\
            0 & 1 & 0 \\
            0 & 0 & 1
        \end{bmatrix}
        \]
    \end{minipage}
    \caption{A hierarchy and its summation matrix. Nodes represent states with values constrained to the sum of their direct children, represented by solid edges. Node indices are chosen according to levels in the hierarchy. The transitive closure of the hierarchy, used to define the summation matrix, is represented by dashed edges.}
    \label{fig:hierarchy_example}
\end{figure}
In the following, we consider a hierarchy as a graph, where nodes represent states of a system and edges define summation constraints between the states. Formally, let $\mathcal{H} = \left(\mathcal{Y}, \mathcal{E}\right)$ be a directed tree where $\mathcal{Y}$ is a finite set of nodes, and $\mathcal{E} \subseteq \mathcal{Y}^2$ is a set of edges directed away from the root node. For nodes $\mu, \eta \in \mathcal{Y}$ we write $\mu \succ \eta$ if $\eta$ is reachable from $\mu$ (e.g. $\mu_1 \succ \mu_3$ in Figure \ref{fig:hierarchy_example}). More formally, let $\mathcal{E}^+$ be the transitive closure of $\mathcal{E}$, then we define the partial ordering $\succ$ as,
\begin{equation*}
    \mu \succ \eta \iff (\mu, \eta) \in \mathcal{E}^+.
\end{equation*}
We say that the interior nodes are the "top level" nodes, and the leaves are the "bottom level" nodes. Let $n_\text{top}$ and $n_\text{bot}$ be the number of top and bottom level nodes, respectively and set $n_\text{tot} = n_\text{top} + n_\text{bot}$ to be the total number of nodes in the hierarchy. Let $y_\text{top} \in \mathbb{R}^{n_\text{top}}$ and $y_\text{bot} \in \mathbb{R}^{n_\text{bot}}$ be the states of the top and bottom level nodes, respectively. We then define the summation matrix $\mathcal{S}$ in terms of its blocks,
\begin{equation*}
    \mathcal{S} = \begin{bmatrix} \mathcal{S}_{\text{top}} \\ I_{n_\text{bot}} \end{bmatrix} \in \mathbb{R}^{n_\text{tot}\times n_\text{bot}},
\end{equation*}
where $I_{n_\text{bot}}$ is the $n_\text{bot} \times n_\text{bot}$ identity matrix. Define the index set $\mathcal{I} = \{1, ..., n_\text{tot}\}$ and assign to each node an index $i \in \mathcal{I}$ such that the first $n_\text{top}$ nodes are top level and the subsequent $n_\text{bot}$ nodes are bottom level. The top level summation matrix $\mathcal{S}_{\text{top}}$ is defined in terms of reachability of the top levels from the bottom levels, i.e.,
\begin{equation*}
    \mathcal{S}_{top,i,j} = \begin{cases}
        1, & \text{if } \text{node } j+n_\text{top} \succ \text{node } i, \\
        0, & \text{otherwise.}
    \end{cases}
\end{equation*}
The constraints of the hierarchy may then be expressed in terms of the linear equation,
\begin{equation}
    \mathcal{S} y_\text{bot} = \begin{pmatrix}
        y_\text{top} \\
        y_\text{bot}
    \end{pmatrix}
    \label{eq:hierarchy_constraint}
\end{equation}
Note, that the ordering of the nodes within the top and bottom levels is so far arbitrary. To provide a convention, we further structure the top levels by their individual level. For a given node $\mu$, define the level $\ell(\mu)$, as the length of the longest directed path to a bottom level node (e.g. $\ell(\mu_1) = 2$ in Figure \ref{fig:hierarchy_example}). We then choose our ordering such that,
\begin{equation*}
    \ell(\mu_i) \geq \ell(\mu_j) \implies i < j.
\end{equation*}
Whilst this ordering is not unique, it satisfies that top and bottom levels are indexed by the first $n_\text{top}$ and last $n_\text{bot}$ indices of $\mathcal{I}$, respectively. Further, if $\mu \succ \eta$ it follows that $\ell(\mu) > \ell(\eta)$, so that parent nodes are ordered above their children.

In the context of forecasting and hierarchical reconciliation, we will consider time dependent states, and assume that observations are available at the bottom level of the hierarchy. The goal of hierarchical reconciliation \parencite{hyndmanOptimalCombinationForecasts2011} may be framed as finding a matrix $\mathcal{P} \in \mathbb{R}^{n_{\text{bot}} \times n_{\text{tot}}}$ and reconciled forecasts $\tilde{y}_t \in \mathbb{R}^{n_{\text{tot}}}$ satisfying,
\begin{equation}
    \label{eq:reconciliation}
    \tilde{y}_t = \mathcal{S} \mathcal{P} \hat{y}_t.
\end{equation}
The matrix $\mathcal{S}\mathcal{P}$ is a projection matrix, mapping base forecasts to a coherent linear subspace \autocite{panagiotelisForecastReconciliationGeometric2021}. In particular, $\mathcal{P}$ maps base forecasts to reconciled bottom level forecasts, which are then aggregated to top level forecasts by the summation matrix $\mathcal{S}$, enforcing coherency by construction. Thus, equation \eqref{eq:reconciliation} expresses the summation constraint \eqref{eq:hierarchy_constraint} in terms of the coherent forecasts $\mathcal{P} \hat{y}_t$ and $\tilde{y}_t$. Various methods have been suggested for determining $\mathcal{P}$ and the reconciled forecasts.
% Bottom up/top down
Simple reconciliation methods include the bottom-up and top-down approaches. In bottom-up reconciliation, the bottom level forecasts are aggregated by summation to form the top level forecasts, whereas the top-down approach distributes top-level forecasts to the bottom levels. These methods ensure coherency, but ignore information available through either top or bottom level forecasts or correlations between these.
% Optimal reconciliation
The optimal combination approach proposed by \cite{hyndmanOptimalCombinationForecasts2011} uses some of this information by posing the problem as linear regression,
\begin{equation*}
        \tilde{y}_t = \mathcal{S} \mathcal{P} \hat{y}_t + \epsilon_t,
        \label{eq:reconiliation_reg_model}
\end{equation*}
where \gls{gls} estimation was suggested to estimate $\mathcal{P}$ as,
\begin{equation*}
    \mathcal{P} = \left( \mathcal{S}^\top \Sigma^\dagger \mathcal{S} \right)^{-1} \mathcal{S}^\top \Sigma^\dagger,
\end{equation*}
where $\Sigma^\dagger$ is the Moore-Penrose generalized inverse of the variance-covariance matrix $\Sigma$ of the reconciliation error $\epsilon_t$. \cite{hyndmanOptimalCombinationForecasts2011} originally used the identity matrix in place of $\Sigma$, thus reducing the problem to \gls{ols}.
Later, \cite{wickramasuriyaOptimalForecastReconciliation2019} showed that $\Sigma$ is unidentifiable and provided theoretical justification for the use of the full variance-covariance matrix of the base forecast errors $W$ as a weight matrix in place of $\Sigma$. The resulting \gls{mint} solution minimizes the trace of the variance of the reconciled forecast errors and reconciliation. Several estimators have been suggested for $W$ \parencite{wickramasuriyaOptimalForecastReconciliation2019,nystrupTemporalHierarchiesAutocorrelation2020}, including shrinkage based estimators \parencite{schaferShrinkageApproachLargeScale2005,ledoitImprovedEstimationCovariance2003}. For temporal hierarchies, autocorrelation structures in the errors may also be used to guide the choice of estimator \parencite{nystrupTemporalHierarchiesAutocorrelation2020}.
The linear model \eqref{eq:reconiliation_reg_model} is stated as a regression problem for the reconciled forecasts. Alternatively, \cite{mollerOptimalForecastReconciliation2024} showed that the problem can be expressed as a \gls{glm} in the bottom level base forecast errors,
\begin{equation}
    \label{eq:hierarchy_glm}
    y_{bot,t} - \hat{y}_{bot,t} = X_t \theta + \epsilon_t, \quad \epsilon_t \sim N(0, \Sigma_\text{r}).
\end{equation}
Where the design matrix is given by,
\begin{equation}
    X_t = \left[ I_{n_\text{bot}} \otimes \left( \hat{y}_{top, t}^\top - \hat{y}_{bot,t}^\top \mathcal{S}_\text{top}^\top \right) \right],
\end{equation}
Here $y_{bot,t}, \hat{y}_{bot,t} \in \mathbb{R}^{n_\text{bot}}$ are the bottom level observations and base forecasts respectively and $\hat{y}_{top,t} \in \mathbb{R}^{n_\text{top}}$ are the top level base forecasts. This form has the advantage, that it enables the use of well known results on linear models. The reconciled forecasts are then given by,
\begin{equation}
    \hat{y}_{t} = \mathcal{S} \left(X_t \theta + \hat{y}_{bot,t}\right).
\end{equation}
\cite{mollerOptimalForecastReconciliation2024} further showed that Tikhonov regularization of the least squares estimate corresponds to shrinkage and \gls{map} estimation. For observations labelled $s=1,...,t$, the response vector and design matrix may be stacked into $y_\text{bot} - \hat{y}_\text{bot} \in \mathbb{R}^{t n_\text{bot}}$ and $X^\text{stack} \in \mathbb{R}^{t n_\text{bot} \times n_\text{top} n_\text{bot}}$. The Tikhonov regularized objective is then given by,
\begin{equation}
    \min_\theta \norm{ X^\text{stack} \theta \ - (y_\text{bot} - \hat{y}_\text{bot}) }_{I_t \otimes \Sigma_\text{r}^{-1}}^2 + \norm{ \theta - \theta_0 }_{\Sigma_{\theta_0}}^2
    \label{eq:tikhonov_hierarchy}
\end{equation}
where $\theta_0$ and $\Sigma_{\theta_0}$ are the prior mean and covariance. Let $\Sigma_\text{top}$ and $\Sigma_\text{bot}$ be the variance-covariance matrices of the top and bottom level base forecast errors, and $\Sigma_\text{top}^d$ and $\Sigma_\text{bot}^d$ be matrices containing their respective diagonals. For the \gls{map} estimate to match shrinkage, \cite{mollerOptimalForecastReconciliation2024} showed that the prior should be chosen as,
\begin{equation}
    \theta_0 = \left(\hat{\Sigma}_\text{top}^d + \mathcal{S}_\text{top} \hat{\Sigma}_\text{bot}^d \mathcal{S}_\text{top}^\top\right)^{-1} \left(\mathcal{S}_\text{top} \hat{\Sigma}_\text{bot}^d\right),
    \label{eq:shrinkage_prior}
\end{equation}
\begin{equation}
    \Sigma_{\theta_0} = \frac{1 - \gamma}{\gamma t} \left( \hat{\Sigma}_\text{top}^d + \mathcal{S}_\text{top} \hat{\Sigma}_\text{bot}^d \mathcal{S}_\text{top}^\top\right)^{-1}.
    \label{eq:shrinkage_prior_cov}
\end{equation}
where $\gamma \in [0,1]$ is the shrinkage parameter.
In this case, the shrinkage estimate of the parameters is given by,
\begin{equation}
    \begin{aligned}
    \hat{\theta} = \left((1-\gamma) \left(\hat{Y}_\text{top} - \hat{Y}_\text{bot} \mathcal{S}_\text{top}^\top\right)^\top \left(\hat{Y}_\text{top} - \hat{Y}_\text{bot} \mathcal{S}_\text{top}^\top\right) + \gamma t \left(\hat{\Sigma}_\text{top}^d + \mathcal{S}_\text{top} \hat{\Sigma}_\text{bot}^d \mathcal{S}_\text{top}^\top\right)\right)^{-1} \\
    \times \left( (1-\gamma) \left(\hat{Y}_\text{top} - \hat{Y}_\text{bot} \mathcal{S}_\text{top}^\top\right) ^\top \left(Y_\text{bot} - \hat{Y}_\text{bot}\right) + \gamma t \mathcal{S}_\text{top} \hat{\Sigma}_\text{bot}^d\right)    
    \end{aligned}
    \label{eq:shrinkage_estimate}
\end{equation} 
where $\hat{Y}_\text{top}, \hat{Y}_\text{bot}, Y_\text{bot}$ are matrices
containing the vectors $\hat{y}_{top,s}, \hat{y}_{bot,s}, y_{bot,s}$ for $s=1,...,t$ as rows. The Tikhonov regularized objective \eqref{eq:tikhonov_hierarchy} and related results provide much of the basis for the online reconciliation methods to be developed in the following sections.
\section{Linear models}\label{sec:GLM}
To formalize a framework for online forecasting and reconciliation, we will use \gls{glm}s as our basic modelling tool. The usual linear model expresses a linear relationship between a target variable and features, assuming i.i.d. Gaussian noise. Instead, we consider a more general form, that allows residuals to be correlated and utilizes matrix normal distributions to characterize distributions in both residuals and parameters. This facilitates a concise and unified representation of the inherently multivariate models and the connections between shrinkage estimators, ridge regression and Bayesian inference. We will use a notation for matrix variate normal distributions as in \cite{glanzExpectationMaximizationAlgorithmMatrix2013}, and refer to the textbook by \cite{gupta2018matrix} for further details, see also Appendix \ref{sec:appendix_useful_identities} for a brief description. Let $x$, $y$ be random variables taking values in $\mathbb{R}^n$ and $\mathbb{R}^m$, respectively. We will consider discrete, regularly sampled time series and use integer indices to denote time. Take a sequence of samples indexed by $s = 1, 2, \ldots, t$ and collect them in design and response matrices $X$ and $Y$, such that the s-th row of $X$ and $Y$ are the row vectors $x_s^\top$ and $y_s^\top$ respectively. We then consider a \gls{glm} of the form,
\begin{equation}
    \label{eq:glm}
    Y = X \theta + E, \quad E \sim \mathcal{MN}(0, U, V)
\end{equation}
where $E$ takes values in $\mathbb{R}^{t \times m}$ and the fixed matrices $U \in \mathbb{R}^{t \times t}$ and $V \in \mathbb{R}^{m \times m}$. This will be the basic model considered in the following. If $U$ has a diagonal structure, the model may equivalently be expressed as a set of $t$ linear models with independent residuals,
\begin{equation}
    y_s = \theta ^\top x_s + e_s, \quad e_s \sim N(0, U_{ss} V), \quad s = 1, 2, \ldots, t.
    \label{eq:independent_glm}
\end{equation}
The latter decomposition will be useful for our development of a recursive inference scheme in section \ref{sec:recursive_estimation}, as the solutions yield themselves to simple incremental updates. More general correlation structures in $U$ pose interesting modelling possibilities, but encumber the solution complexity and are thus only considered in the non-recursive setting. For $U$ as the identity matrix, the model reduces to a standard multivariate linear model with i.i.d. residuals.
\subsection{Regularization in linear models}\label{sec:GLM_regularization}
We pose parameter estimation as a ridge regression problem with the objective,
\begin{equation}
    J(\theta; P, Q, R, \theta_0) = \norm{X \theta - Y}_{P, R}^2 + \norm{\theta - \theta_0}_{Q, R}^2.
    \label{eq:ridge_regression_objective}
\end{equation}
where $\theta_0$ is a shrinkage target, $P \in \mathbb{R}^{t \times t}$ a weight matrix, $Q \in \mathbb{R}^{n \times n}$ a regularization parameter and $R \in \mathbb{R}^{m \times m}$. Here we have introduced a weighted Mahalanobis-type norm for matrices, which we define as,
% Perhaps this should be refered to as a Mahalanobis-type norm?
\begin{equation}
    \norm{X}_{A,B}^2 = \text{tr}\left(B X^\top A X \right).
    \label{eq:weighted_matrix_norm}
\end{equation}
It holds that $\norm{X}_{A,B}^2 = \vect{X}^\top (B \otimes A) \vect{X}$, so that the norm is well-defined if $B \otimes A$ is \gls{spd}. This also implies that $J$ is convex in $\theta$, if $R \otimes P$ and $R \otimes Q$ are both \gls{spd}.
%\begin{equation}
%    \norm{A}_{F,W}^2 = \text{tr}\left( A^\top W A \right).
%\end{equation}
Loosely speaking, the objective seeks to minimize the weighted variance of the prediction error, while penalizing deviations of the parameters from the shrinkage target $\theta_0$. In the following we will assume that $P$, $Q$ and $R$ are all \gls{spd}. It now holds, that the minimizer of \eqref{eq:ridge_regression_objective} with respect to $\theta$ is independent of $R$ and has a closed form solution. To state the result, we introduce the notation,
\begin{equation*}
    K = X^\top P X, \quad L = X^\top P Y, \quad H = X^\top P U P X.
\end{equation*}
Minimizing \eqref{eq:ridge_regression_objective}, the parameter estimate can be written,
\begin{equation}
    \hat{\theta} = \left(K + Q\right)^{-1} \left( L + Q \theta_0 \right).
    \label{eq:ridge_regression_solution}
\end{equation}
Under the model \eqref{eq:glm}, the distribution of the estimate is matrix normal with,
\begin{equation}
    \label{eq:ridge_regression_solution_distribution}
    \hat{\theta} \sim \mathcal{MN} \left(\theta_\text{ridge}, \Psi_{\text{ridge}}, V \right),
\end{equation}
where the mean and covariance parameters are given by,
\begin{equation*}
    \theta_\text{ridge} = \left(K + Q\right)^{-1} \left(X^\top P X \theta + Q \theta_0 \right), \quad
    \Psi_{\text{ridge}} = \left(K + Q\right)^{-1} H \left(K + Q\right)^{-1}.
\end{equation*}
Note, the mean estimate $\hat{\theta}$ is equivalent to the \gls{map} estimate of $\theta$ under the model \eqref{eq:glm}, when assuming $U, V$ known and prior distribution,
\begin{equation*}
    \theta \sim \mathcal{MN}(\theta_0, Q^{-1}, R^{-1}).
\end{equation*}
For a full discussion of the Bayesian interpretation of the model and a proof of the above statement, see Appendix \ref{sec:appendix_bayesian_regression}. Given the results in the Bayesian setting, it is straightforward to argue that $\hat{\theta}$ is the minimizer of \eqref{eq:ridge_regression_objective}, see Appendix \ref{sec:appendix_ridge_regression}. The predictive distribution for a new observation $Y^*$ given new features $X^*$ is then given by,
\begin{equation}
    Y^* | X^*, X, Y \sim \mathcal{MN}\left( X^* \theta_\text{ridge}, X^* \Psi_{\text{ridge}} X^{* \top} + U^*, V \right),
    \label{eq:ridge_regression_predictive_distribution_main}
\end{equation}
It is worthwhile to consider some variants of the model \eqref{eq:glm} and their relation to standard methods. Firstly, setting $Q = 0$, it is easy to see that the minimization of \eqref{eq:ridge_regression_objective} and the solution \eqref{eq:ridge_regression_solution} reduces to a \gls{wls} type problem and that the parameter estimate becomes unbiased. More generally, the model \eqref{eq:glm} may be expressed in the vectorized form,
\begin{equation*}
    \vect{Y} = (I_m \otimes X) \vect{\theta} + \vect{E}, \quad \vect{E} \sim N(0, V \otimes U),
\end{equation*}
where $I_m$ is the identity matrix of size $m$. This expression conforms to the standard linear model form, directly allowing the use of associated theory and methods to recover estimates of the parameter and predictive distributions. The major drawback of this approach is the additional computational complexity implied by the larger design and covariance matrices. These issues could potentially be alleviated by exploiting the sparse structures, especially in the case of $U$ being equal to the identity or other sparse matrix. In the vectorized form, diagonal $U$ corresponds to,
\begin{equation*}
    y_s = (I_m \otimes X) \vect{\theta} + e_s, \quad e_s \sim N(0, u_{ss} V).
\end{equation*}
For $U = I_m$, this is exactly the formulation used for the hierarchical \gls{glm} by \cite{mollerOptimalForecastReconciliation2024}, i.e. equation \eqref{eq:hierarchy_glm}. They suggest that this formulation may be used for model reduction using a test strategy to remove insignificant parameters for hierarchical reconciliation. This corresponds to removing individual columns from the design matrix $I_m \otimes X$, which is not possible in the matrix form \eqref{eq:glm} without adding algebraic constraints to the optimization problem. Despite the flexibility of the vectorized formulation, we will use the matrix form \eqref{eq:glm} for its notational simplicity and computational efficiency.

\subsection{Recursive Estimation}\label{sec:recursive_estimation}
% Out of sample, simply remove the last term (no covariance between parameter and future observation noise)
For use in an online setting, we assume our inputs and outputs are time series with equidistant sampling times. To emphasize the temporal nature of the data, we will attach a subscript $s = 1, 2, \ldots, t$ to all variables, where $t > 0$ is the current time step or simply the (current) time. For example, $y_t$ is the most current observation and $Y_t$ is the collection of all observations currently available from time $1$ to time $t$. At time $t$ we are interested in forecasting $y_{t+k}$ for some integer horizon $k > 0$. For any such $t$, forecasts should be made online, i.e. using models estimated from all data available up to that time. We then consider a sequence of linear models of the form \eqref{eq:glm} parameterized by the current time $t$,
\begin{equation}
    Y_t = X_t \theta + E_t, \quad E_t \sim MN(0, U_t, V), \quad t = 1, 2, \ldots
    \label{eq:glm_t}
\end{equation}
As new observations become available, the matrices $X_t$, $Y_t$, $U_t$ and regularization parameter $P_t$ increase in size, so that numerics eventually become infeasible. To address this, we adopt a recursive approach (see e.g. \parencite{madsenTimeSeriesAnalysis2007,bacherOnlineforecastPackageAdaptive2023}). We do not pursue a general solution in this case, but assume special structures of $P_t$ and $U_t$. Using the subscript notation, we can express the ridge regression estimator \eqref{eq:ridge_regression_solution_distribution} at time $t$ as,
\begin{equation}
    \theta_{\text{ridge}, t} = \left( K_t + Q \right)^{-1} \left( L_t + Q \theta_0 \right),
    \label{eq:ridge_regression_solution_t}
\end{equation} 
and the covariance of the estimator as,
\begin{equation}
    \Psi_{\text{ridge}, t} = \left( K_t + Q \right)^{-1} X_t^\top H_t \left( K_t + Q \right)^{-1}.
    \label{eq:ridge_regression_covariance_t}
\end{equation}
We now assume independent residuals by setting $U_t = \text{diag}(u_1, \ldots, u_t)$, and choose $P_t$ as a diagonal weight matrix,
\begin{equation*}
    \Lambda_t = \text{diag}(\lambda^{t-1}, \lambda^{t-2}, \ldots, 1).
\end{equation*}
For implementation purposes, the main result is recursive update equations for the matrices $K_t$, $L_t$ and $H_t$.
\begin{corollary}[Update equations]\label{cor:recursive_updates}
Under the above assumptions, the ridge regression estimator \eqref{eq:ridge_regression_solution_t} and its covariance \eqref{eq:ridge_regression_covariance_t} can be updated recursively in terms of the matrices $K_t$, $L_t$ and $H_t$ as,
\begin{equation}
    \begin{aligned}
        K_{t+1} &= \lambda K_t + x_{t+1} x_{t+1}^\top, \\
        L_{t+1} &= \lambda L_t + x_{t+1} y_{t+1}^\top, \\
        H_{t+1} &= \lambda^2 H_t + u_{t+1} x_{t+1} x_{t+1}^\top.
    \end{aligned}
    \label{eq:ridge_regression_recursive_updates}
\end{equation}
\end{corollary}
\noindent The corollary follows from Proposition \ref{prop:bayesian_posterior} in the appendix, and is proven in Appendix \ref{sec:appendix_recursive}. The choice of $P_t = \Lambda_t$ is motivated by commonly used exponential smoothing filters. Indeed, for univariate target variable $Y_t$, setting $Q_t=0$ and $u_t = 1$, the above reduces to the recursive least squares solution used by e.g. \cite{bacherOnlineforecastPackageAdaptive2023}. For a forecast horizon $h > 0$, the predictive distribution for time $t+h$ is then given by,
\begin{equation*}
    y_{t+h} \mid Y_t \sim \mathcal{N} \left( \hat{\theta}_{\text{ridge}, t} x^*_{t+h}, (x^{* \top}_{t+h} \Psi_{\text{ridge}, t} x^*_{t+h} + 1) V \right). 
    \label{eq:ridge_regression_recursive_prediction}
\end{equation*}
As we will typically not know the variance $V$, it should be estimated from the data. For example, we may use the weighted empirical variance estimate at time $t$,
\begin{equation*}
    \hat{V_t} = (X_t \hat{\theta}_t - Y_t)^\top \Lambda_t (X_t \hat{\theta}_t - Y_t).
\end{equation*}
Assuming large $t$, the memory is saturated, and the variance update may be expressed recursively as \parencite{bergsteinssonHeatLoadForecasting2021},
\begin{equation}
    \hat{V}_{t+1} = \lambda \hat{V}_t + (1- \lambda) (\hat{\theta}_t^\top x_{t+1} - y_{t+1})(\hat{\theta}_t^\top x_{t+1} - y_{t+1})^\top.
    \label{eq:pred_err_var_recursive}
\end{equation}
We do not pursue a correction to the uncertainty of the predictive distribution implied by the estimation of $V$, but note its potential importance for small samples or short memory.

\subsection{Application to hierarchies}\label{sec:GLM_hierarchies}
Hierarchical forecast reconciliation can be treated using the methods discussed so far by applying the result of \cite{mollerOptimalForecastReconciliation2024}. To do so, we organize the top and bottom level forecasts $\hat{y}_{\text{bot}, s}$ and $\hat{y}_{\text{top}, s}$ and the bottom level observations $y_{\text{bot}, s}$ for $s = 1, 2, \ldots, t$ as rows in the matrices $\hat{Y}_\text{top}$, $\hat{Y}_\text{bot}$ and $Y_\text{bot}$. We then pose a linear model for the bottom level base forecast errors of the sort \eqref{eq:glm},
\begin{equation}
    Y_\text{bot} - \hat{Y}_\text{bot} = (\hat{Y}_\text{top} - \hat{Y}_\text{bot} \mathcal{S}_\text{top}^\top) \theta + E, \quad E \sim \mathcal{MN}(0, U, \Sigma_r).
    \label{eq:matrix_hierarchy_glm}
\end{equation}
Parameters are estimated by minimizing the equivalent of \eqref{eq:ridge_regression_objective} for the model \eqref{eq:matrix_hierarchy_glm},
\begin{equation}
    \norm{(\hat{Y}_\text{top} - \hat{Y}_\text{bot} \mathcal{S}_\text{top}^\top) \theta - (Y_\text{bot} - \hat{Y}_\text{bot})}_{P, \Sigma_r^{-1}}^2 + \norm{\theta - \theta_0}_{Q, \Sigma_r^{-1}}^2.
    \label{eq:hierarchical_obj}
\end{equation}
For $P = U = I_t$, the objective reduces to \eqref{eq:tikhonov_hierarchy}, so that estimates correspond to usual reconciliation methods. For instance, choosing $Q = \Sigma_{\theta_0}^{-1}$ and $\theta_0$ as in \eqref{eq:shrinkage_prior}-\eqref{eq:shrinkage_prior_cov} we recover the shrinkage estimate \eqref{eq:shrinkage_estimate} as the solution \eqref{eq:ridge_regression_solution}. To see the equivalence with \eqref{eq:hierarchy_glm}, we note that when $U = I_t$, a single sample of \eqref{eq:matrix_hierarchy_glm} reduces to,
\begin{equation*}
    \begin{aligned}
        \left(y_{\text{bot}, s} - \hat{y}_{\text{bot}, s}\right)^\top &= \left(\hat{y}_{\text{top}, s} -  S_{\text{top}}\hat{y}_{\text{bot}, s}\right)^\top \theta + e^\top, \\
        &= I_m \otimes \left(\hat{y}_{\text{top}, s} - S_{\text{top}}\hat{y}_{\text{bot}, s}\right) \vect{\theta} + e^\top,  \quad e \sim N(0, \Sigma_r).
    \end{aligned}
\end{equation*}
The model \eqref{eq:matrix_hierarchy_glm} produces predictions $\tilde{\Delta}_\text{bot} = \Delta_\text{bot} \mid Y_\text{bot}$ of the bottom level base forecast errors $\Delta_\text{bot} = Y_\text{bot} - \hat{Y}_\text{bot}$.
Bottom level reconciled forecasts are obtained as $\tilde{Y}_\text{bot} = \hat{Y}_\text{bot} + \tilde{\Delta}_\text{bot}$, or equivalently by computing $\tilde{Y}_\text{bot} = \hat{Y} \mathcal{P}(\theta)^\top$ where the matrix $\mathcal{P}(\theta)$ is defined by \cite{mollerOptimalForecastReconciliation2024} as,
\begin{equation*}
    \label{eq:hierarchy_projection}
    \mathcal{P}(\theta)^\top = \begin{bmatrix} \theta^\top \\  I_{n_{\text{bot}}} - S_{top}^\top \theta^\top \end{bmatrix}.
\end{equation*}
The first term of the objective \eqref{eq:hierarchical_obj} can then be expressed as,
\begin{equation*}
    \norm{\hat{Y}\mathcal{P}(\theta)^\top - Y_{bot}}_{P, \Sigma_r^{-1}}^2.
\end{equation*}
It is then clear, that the solution $\hat{\theta}$ should minimize the bottom level reconciled forecast error, subject to the summation constraint imposed by $\mathcal{P}(\theta)$ and the regularization penalties on the parameters. In particular, setting $P=I_t$ and identifying $Q = \Sigma_{\theta_0}^{-1}$ and $\theta_0$ as in \eqref{eq:shrinkage_prior}-\eqref{eq:shrinkage_prior_cov} we recover the shrinkage estimate \eqref{eq:shrinkage_estimate} as the solution \eqref{eq:ridge_regression_solution}. This is simply restating the findings of \cite{mollerOptimalForecastReconciliation2024}, here based on a matrix variate model. In the following examples however, we do not pursue this exact choice of prior, but consider alternatives that are simpler in this setting. As a special case, we consider $\theta_0 = 0$ and $Q = q I_{n_\text{top}}$ where $q \rightarrow \infty$. In the limit, the prior approaches a point mass at zero, so that in turn the parameters $\theta$ tend to zero, corresponding to a zero prediction of the bottom level forecast errors. The resulting reconciled forecasts will then be the bottom-up forecasts. Note, whilst this prior is common for ridge regression, it differs from the usual shrinkage prior used for forecast reconciliation \parencite{nystrupTemporalHierarchiesAutocorrelation2020}. Since strong regularization under this prior favours the bottom-up forecasts, it is important to consider whether the hierarchy's nodes of interest will benefit from the choice. Using \eqref{eq:matrix_hierarchy_glm}-\eqref{eq:hierarchical_obj}, it is possible to follow either a Bayesian or ridge regression approach by imposing appropriate assumptions on $P$, $U$, $Q$ and $\theta$. As we aim to use the model for online forecast reconciliation, we will pursue the ridge regression formulation and recursive estimation scheme described in the previous sections. Assume then that $\hat{\theta}$ was found using \eqref{eq:ridge_regression_solution} and that we are given new base forecasts $\hat{Y}_\text{bot}^*$, $\hat{Y}_\text{top}^*$, error row covariance $U^*$ and need to predict deviations from the yet unknown output $Y_\text{bot}^*$. Using the input features $X^* = \hat{Y}_\text{top}^* - \hat{Y}_\text{bot}^* \mathcal{S}_\text{top}^\top$, \eqref{eq:ridge_regression_predictive_distribution_main} gives a predictive distribution for $\tilde{\Delta}_\text{bot}^*$, i.e. the conditional distribution of the corresponding bottom level base forecast errors $\tilde{\Delta}_\text{bot}^* = (Y_\text{bot}^* - \hat{Y}_\text{bot}^*)$.
The full reconciled forecast may be obtained by shifting and aggregating the prediction,
\begin{equation*}
\tilde{Y}^* = (\hat{Y}_\text{bot}^* + \tilde{\Delta}_\text{bot}^*) \mathcal{S}^\top.
\end{equation*}
The distribution of the reconciled forecasts follows from the properties of the matrix normal distribution as,
\begin{equation*}
    \tilde{Y}^* \sim \mathcal{MN} \left((\hat{Y}_\text{bot}^* + X^* \hat{\theta}_\text{ridge}) \mathcal{S}^\top, X^* \Psi_{\text{ridge}} X^{* \top} + U^*, \mathcal{S} \Sigma_r \mathcal{S}^\top \right).
\end{equation*}
For a single new sample at time $t+1$, the row covariance reduces to a scalar value dependent on $u_{t+1} = U_{t+1,t+1}$, and by letting $x_{t+1} = \hat{y}_{\text{top}, t+1} - \mathcal{S}_\text{top} \hat{y}_{\text{bot}, t+1}$, the predictive distribution of the reconciled forecast collapses to a multivariate normal distribution,
\begin{equation*}
    \tilde{y}_{t+1} \mid Y_t \sim N \left(\mathcal{S} (\hat{y}_{\text{bot}, t+1} + \hat{\theta}_\text{ridge} x_{t+1}), (x_{t+1}^{*\top} \psi_\text{ridge} x_{t+1}^* + u_{t+1}^*) \mathcal{S} \Sigma_r \mathcal{S}^\top \right).
    \label{eq:hierarchy_predictive_distribution_singleton}
\end{equation*}
Note, when the distribution of the reconciled forecasts are only required per sample, both mean and covariance can be retrieved from the per-sample distribution of the bottom level base forecast errors. In particular for the covariance, it is sufficient to track the combined covariance matrix of the bottom levels $\Sigma_\text{bot} =(x_{t+1}^{*\top} \psi_\text{ridge} x_{t+1}^* + u_{t+1}) \Sigma_r$ in order to recover the reconciled forecast variance.
\section{Implementation}\label{sec:implementation}
Hierarchical forecast reconciliation was implemented in the \textit{PyOnlineForecast} package \parencite{Ronlev-Knudsen_PyOnlineForecast_2026} using the models described in sections \ref{sec:recursive_estimation}-\ref{sec:GLM_hierarchies}. The package employs a notion of \textit{transformations}, which define modular computational procedures that can be applied to data. Simple transformations may be composed to create more complex procedures or composite transformations. Forecast reconciliation can be viewed as one such composite transformation, mapping base forecasts and bottom level observations to reconciled forecasts. The implementation is split into 
\begin{enumerate}
    \item a generic prediction model using the recursive approach described in \ref{sec:recursive_estimation} and 
    \item a set of utilities based on \ref{sec:GLM_hierarchies} to accommodate hierarchical forecast data for use with the regression model and subsequent construction of reconciled forecasts.
\end{enumerate}
The regression model is implemented under the assumption that $U = I$ using variance estimation for $V$ as in \eqref{eq:pred_err_var_recursive} and provides estimates of the per-sample predictive distribution of forecasts. These computations are natively restricted to be performed one sample at a time and are implemented accordingly. This is not true for the construction of reconciled forecasts from the regression model outputs, which can be performed as batch operations, also when only providing per-sample predictive distributions. In the online setting, forecasts and observations for a given time are not available at the same time. Instead, at time $t$ we collect new observations $y_t$ and forecasts $\hat{y}_{t+h}$ for some integer forecast horizon $h > 0$. To construct inputs for the regression model, forecasts should be aligned to observations by back-shifting the forecasts $h$ time steps. Similarly, for training the regression model to that same horizon, features should be lagged $h$ steps relative to the target variable. The latter is automatically handled by the implementation of the ridge regression model transformation, whilst the former is resolved within the reconciliation transformation. The implementation in the \textit{PyOnlineForecast} package is adequately represented by a computational graph of the intermediate transformations. 

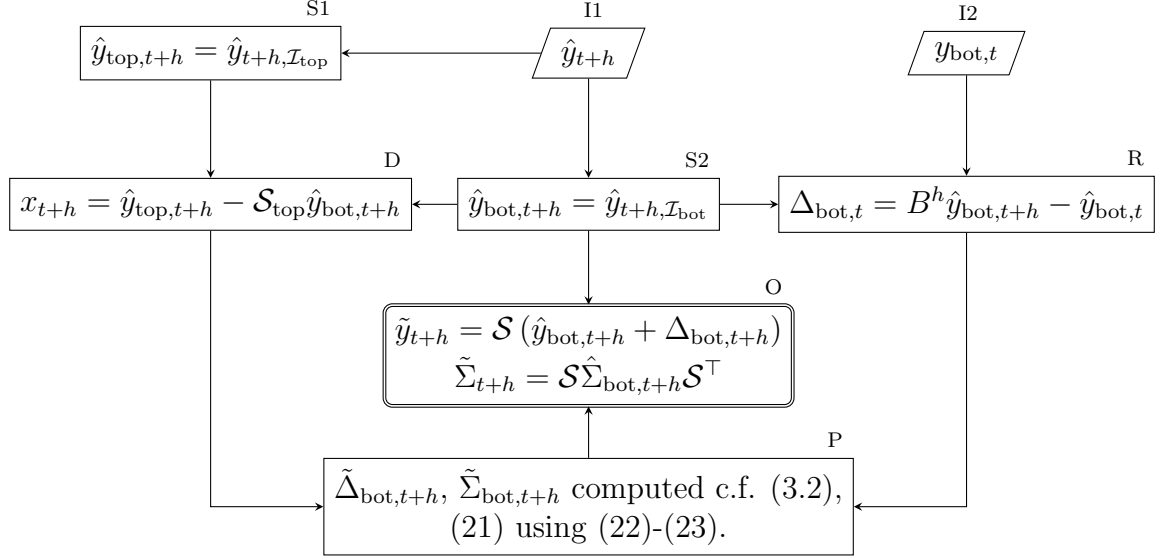
\begin{figure}[h]
\centering
\begin{tikzpicture}[
    node distance=1.5cm and 2.5cm,
    box/.style={rectangle, draw, align=center},
    >=stealth
]
    % Define node contents
    \def\contentYHatTop{$\hat{y}_{\text{top}, t+h} = \hat{y}_{t+h, \mathcal{I}_\text{top}}$}
    \def\contentYHat{$\hat{y}_{t+h}$}
    \def\contentYBot{$y_{\text{bot}, t}$}
    \def\contentX{$x_{t+h} = \hat{y}_{\text{top}, t+h} - \mathcal{S}_{\text{top}} \hat{y}_{\text{bot}, t+h}$}
    \def\contentYHatBot{$\hat{y}_{\text{bot}, t+h} = \hat{y}_{t+h, \mathcal{I}_\text{bot}}$}
    \def\contentY{$\Delta_{\text{bot}, t} = B^h \hat{y}_{\text{bot},t+h} - \hat{y}_{\text{bot}, t}$}
    \def\contentOutput{\shortstack{$\tilde{y}_{t+h} = \mathcal{S} \left(\hat{y}_{\text{bot},t+h} + \Delta_{\text{bot}, t+h}\right)$\\$\tilde{\Sigma}_{t+h} = \mathcal{S} \hat{\Sigma}_{\text{bot}, t+h} \mathcal{S} ^\top$}}
    \def\contentRRR{$\tilde{\Delta}_{\text{bot}, t+h}$, $\tilde{\Sigma}_{\text{bot}, t+h}$ computed c.f. \eqref{eq:ridge_regression_recursive_prediction}, \eqref{eq:pred_err_var_recursive} using \eqref{eq:matrix_hierarchy_glm}-\eqref{eq:hierarchical_obj}.}
    
    % Fixed grid with explicit positioning: columns at 0cm, 2.5cm, 5cm; rows at 0cm, -1.5cm, -3cm, -4.5cm
    \node[box] (YHatTop) at (0cm, 0cm) {\contentYHatTop};
    \node[draw, trapezium, trapezium left angle=70, trapezium right angle=110, align=center] (YHat) at (5cm, 0cm) {\contentYHat};
    \node[draw, trapezium, trapezium left angle=70, trapezium right angle=110, align=center] (YBot) at (10cm, 0cm) {\contentYBot};
    
    \node[box] (X) at (0cm, -2cm) {\contentX};
    \node[box] (YHatBot) at (5cm, -2cm) {\contentYHatBot};
    \node[box] (Y) at (10cm, -2cm) {\contentY};
    
    \node[draw, rectangle, rounded corners=3pt, double, align=center] (Output) at (5cm, -4cm) {\contentOutput};
    
    \node[box] (RRR) at (5cm, -6cm) {\begin{varwidth}{7cm}\centering \contentRRR \end{varwidth}};

    % Node identifiers
    \node[anchor=south east, font=\scriptsize] at (YHat.north east) {I1};
    \node[anchor=south east, font=\scriptsize] at (YBot.north east) {I2};
    \node[anchor=south east, font=\scriptsize] at (YHatTop.north east) {S1};
    \node[anchor=south east, font=\scriptsize] at (YHatBot.north east) {S2};
    \node[anchor=south east, font=\scriptsize] at (X.north east) {D};
    \node[anchor=south east, font=\scriptsize] at (Y.north east) {R};
    \node[anchor=south east, font=\scriptsize] at (RRR.north east) {P};
    \node[anchor=south east, font=\scriptsize] at (Output.north east) {O};

    % Edges
    \draw[->] (YHat) -- (YHatTop);
    \draw[->] (YHat) -- (YHatBot);

    \draw[->] (YHatBot) -- (X);
    \draw[->] (YHatTop) -- (X);

    \draw[->] (YBot) -- (Y);
    \draw[->] (YHatBot) -- (Y);

    \draw[->] (X) |- (RRR);
    \draw[->] (Y) |- (RRR);

    \draw[->] (RRR) -- (Output);
    \draw[->] (YHatBot) -- (Output);
\end{tikzpicture}
\caption{Schematic of the implementation of forecast reconciliation in the \textit{PyOnlineForecast} package. Rectangles represent transformations with arrows indicating data flow. Inputs are drawn using trapezoids and the output computation is highlighted with double edges and rounded corners. I1-I2: input data. S1-S2: splitting of base forecasts into top and bottom levels. D: construction of design matrix from base forecasts. R: back-shifting of bottom level base forecasts and construction of response matrix from bottom level base forecast errors. P: prediction of bottom level base forecast errors using the recursive ridge regression transformation. O: construction of reconciled forecasts for transformation output.}
\label{fig:transformation_graph}
\end{figure}

Figure \ref{fig:transformation_graph} illustrates the most important transformation compositions and output computations used to define the reconciliation transformation. At runtime, the transformation receives base forecasts $\hat{y}_{t+h}$ and bottom level observations $y_{\text{bot}, t}$ as inputs. The base forecasts are split into top and bottom levels by selecting indices  $\mathcal{I}_\text{top} = \{1, \ldots, n_\text{top}\}$ and $\mathcal{I}_\text{bot} = \{n_\text{top} + 1, \ldots, n_\text{top} + n_\text{bot}\}$, respectively. For computing the response matrix, the bottom level base forecasts are back-shifted, written as $B^h \hat{y}_{\text{bot}, t+h} = \hat{y}_{\text{bot}, t}$. The design and response matrices are then constructed based on current and past data. The ridge regression module (Figure \ref{fig:transformation_graph}, box P) uses these matrices as inputs, and internally handles back-shifting, training and prediction of the bottom level forecast errors. The module further supports estimation of covariance matrix $V$ using equation \eqref{eq:pred_err_var_recursive}, but does not incorporate estimation uncertainty of this parameter into the predictive distribution. The predictions are then used by the reconciliation transformation to construct reconciled forecasts which are returned as output.

The implementation can be extended to temporal hierarchies by preprocessing inputs before applying the reconciliation transformation. This step constructs bottom level observations for the temporal hierarchy based on a single observed time series. The revised input is formed by constructing additional input columns of back-shifted observations, one for each node in the bottom level of the temporal hierarchy. The temporal reconciliation procedure is then specified by defining the summation matrix as before, and a back-shift structure for the bottom level observations. For an input variable $z_t$ taking values in $\mathbb{R}^n$, the back-shifted bottom level observations are constructed using the operator $\mathcal{B}$ defined as,
\begin{equation*}
    y_{\text{bot}, t} = \mathcal{B} z_t = \left(B^{j_1} z_{i_1, t}, \ldots, B^{j_m} z_{i_m, t}\right)^\top = \left(z_{i_1, t-j_1}, \ldots, z_{i_m, t-j_m}\right)^\top.
\end{equation*}
In practice the operator can be specified as a set of pairs $(i_k, j_k)$ for $k = 1, \ldots, m$. Here $1 < i_k < n$ and $0 < j_k$ are the indices of the input and a corresponding lag. The operator effectively maps state observations to a multivariate time series obeying the temporal hierarchy structure. This approach comes with the caveat that some observations will be repeated in the bottom level observations $y_{\text{bot}, t}$, when they are updated at every time step $t$. This may artificially introduce a strong autocorrelation structure in the data. To mitigate this issue, the implementation accepts a boolean argument controlling the update frequency. When enabled, parameters are only updated every $h$-th time step, where $h$ is the maximum of the input lags.
\section{Adaptive heat load forecasting}\label{sec:adaptive_heat_load_forecasting}
We demonstrate the proposed methods by forecasting district heating loads for three areas in the Gellerup district of Aarhus, Denmark. Data was provided by the district heating operator Kredsløb. The data consists of three years of heat load measurements from March 2022 to April 2025. The measurements were recorded at a local heat exchanger supplying heat to three different regions; south, east and west. The measurements were provided on a five-minute frequency. Pre-processing of the data includes removal of outliers, resampling to hourly frequency using averaging and filling empty values using persistence (forward fill). Outlier removal was performed by assigning standard scores to the observations, and discarding observations with high absolute scores. The standard scores were computed using a running estimate of the mean and variance with an exponential forgetting factor of $0.995$, and the threshold for removal was set to three standard deviations. Figure \ref{fig:load_timeseries} shows the cleaned time series for each of the three regions.
\begin{figure}
    \centering
    \includegraphics{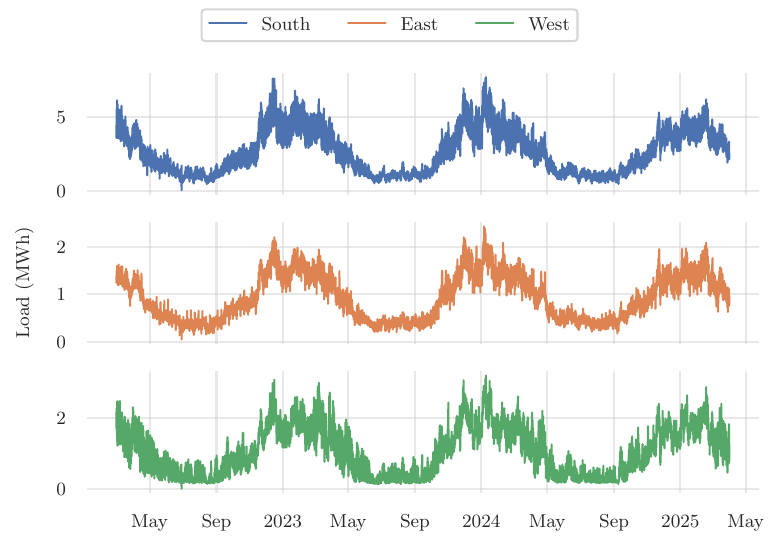}
    \caption{Heat load data provided by Kredsløb after removing outliers and resampling to hourly frequency.}
    \label{fig:load_timeseries}
\end{figure}
% Hierarchy
A temporal hierarchy is constructed for a daily cycle consisting of 24 hourly forecasts at the bottom level and aggregation levels 6, 12, and 24 hours. Figure \ref{fig:hierarchy_24h} illustrates the structure of the hierarchy. Throughout the results, performance metrics are computed based on residuals after a burn-in period of 90 days.
\begin{figure}
\begin{center}
\begin{tikzpicture}[
    node distance=1.5cm,
    every node/.style={draw, circle, minimum size=0.3cm, font=\tiny}
]

    \draw[thick] (0, -2.8) -- (9.6, -2.8);
    
    % Add tick marks at positions 1-24
    \foreach \i in {1,2,...,24} {
        \draw[thick] (\i*0.4-0.2, -2.8) -- (\i*0.4-0.2, -2.7);
    }
    
    % Add horizon labels
    \node[below, font=\small, draw=none, circle=false] at (4.8, -2.8) {Horizon};
    \node[below, font=\tiny, draw=none, circle=false] at (0.2, -2.9) {1h};
    \node[below, font=\tiny, draw=none, circle=false] at (9.4, -2.9) {24h};
    
    % Add aggregation level labels on the left
    \node[left, font=\tiny, draw=none, circle=false] at (-0.3, -2) {1h};
    \node[left, font=\tiny, draw=none, circle=false] at (-0.3, -1) {6h};
    \node[left, font=\tiny, draw=none, circle=false] at (-0.3, 0) {12h};
    \node[left, font=\tiny, draw=none, circle=false] at (-0.3, 1) {24h};
    
    % Bottom level (24 nodes)
    \foreach \i in {1,2,...,24} {
        \node (b\i) [align=center, inner sep=0pt,minimum size = 0.35cm] at (\i*0.4-0.2, -2) {\fontsize{5}{5}\selectfont\i};
    }
    
    % Second level (4 nodes at positions 6, 12, 18, 24)
    \node (s1)[inner sep=0pt,minimum size = 0.6cm] at (2.2, -1) {\fontsize{5}{5}\selectfont 1-6};  % position 6
    \node (s2)[inner sep=0pt,minimum size = 0.6cm] at (4.6, -1) {\fontsize{5}{5}\selectfont 7-12};  % position 12
    \node (s3)[inner sep=0pt,minimum size = 0.6cm] at (7.0, -1) {\fontsize{5}{5}\selectfont 13-18};  % position 18
    \node (s4)[inner sep=0pt,minimum size = 0.6cm] at (9.4, -1) {\fontsize{5}{5}\selectfont 19-24};  % position 24
    
    % Third level (2 nodes at positions 12, 24)
    \node (t1)[inner sep=0pt,minimum size = 0.7cm] at (4.6, 0) {\fontsize{5}{5}\selectfont 1-12};   % position 12
    \node (t2)[inner sep=0pt,minimum size = 0.7cm] at (9.4, 0) {\fontsize{5}{5}\selectfont 12-24};   % position 24
    
    % Top level (1 node at position 24)
    \node (root)[inner sep=0pt,minimum size = 0.8cm] at (9.4, 1) {\fontsize{5}{5}\selectfont 1-24};
    
    % Connections from bottom to second level (groups of 6)
    \foreach \i in {1,2,...,6} {
        \draw (b\i) -- (s1);
    }
    \foreach \i in {7,8,...,12} {
        \draw (b\i) -- (s2);
    }
    \foreach \i in {13,14,...,18} {
        \draw (b\i) -- (s3);
    }
    \foreach \i in {19,20,...,24} {
        \draw (b\i) -- (s4);
    }
    
    % Connections from second to third level (groups of 2)
    \draw (s1) -- (t1);
    \draw (s2) -- (t1);
    \draw (s3) -- (t2);
    \draw (s4) -- (t2);
    
    % Connections from third to top level
    \draw (t1) -- (root);
    \draw (t2) -- (root);
\end{tikzpicture}
\end{center}
\caption{Temporal hierarchy used for load forecasting. The bottom level considers horizons 1 through 24 hours, aggregated into 6-hour, 12-hour, and 24-hour levels.}
\label{fig:hierarchy_24h}
\end{figure}
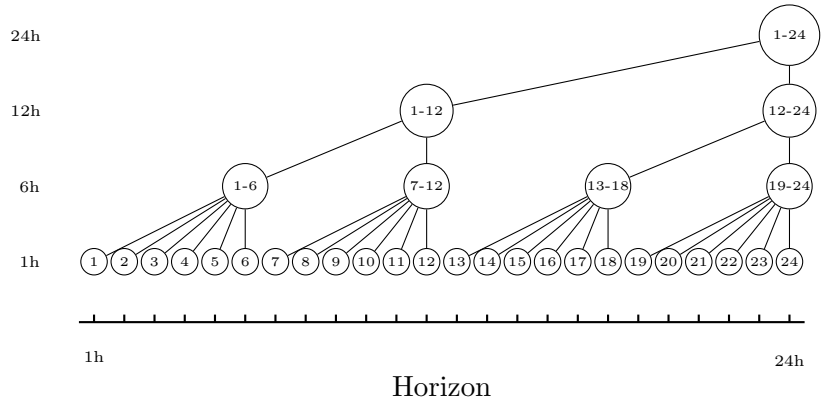

% Base forecasts
\paragraph{Base forecasts.}
Base forecasts were generated for all three regions and each node in the hierarchy using the \textit{PyOnlineForecast} package. Each node in the hierarchy corresponds to a separate forecast model using a number of input features derived from data transformations of observations and forecasts for the given horizon. As inputs, the models consider previous load measurements, local temperature measurements and weather forecasts. The local temperature measurements were provided by Kredsløb, originating from a nearby weather station, see Figure \ref{fig:temperature_timeseries}. This data was provided on an hourly resolution, and was pre-processed by removing outliers and filling missing values as described for the load data. Weather forecasts were obtained from the MET Norway THREDDS data server \parencite{metnorway_thredds}, extracted for a nearby geographical location. Each weather forecast provides hourly predictions for the following 58 hours and was updated every six hours. To align with load data, forecasts were resampled to hourly frequency by reusing forecasts between updates. This way, 24 hour-ahead forecasts were made available every hour. For any hour in between the six-hourly updates, the forecast for horizon $h$ is simply the most recent forecast of horizon $h+k$ where $k$ is the number of hours since the last update. Feature selection and design was done by hand, with hyperparameters being tuned for the east region and applied to all three regions. A linear model based on section \ref{sec:recursive_estimation} was used to estimate parameters and generate forecasts. For further details on the base forecasts, refer to Appendix \ref{sec:appendix_base_models}.
% - Temperature data
% - Weather forecasts
\begin{figure}
    \centering
    \includegraphics{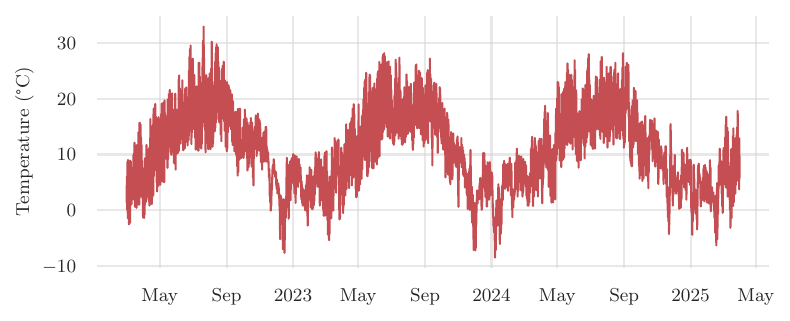}
    \caption{Local temperature measurements.}
    \label{fig:temperature_timeseries}
\end{figure}
\begin{table}[h]
\centering
\begin{tabular}{ccccc}
\toprule
\textbf{Case} & Update frequency & $\theta_0$ & $Q$ & $\lambda_\text{mem}$ \\
\midrule
Main & Hourly & 0 & 0.001 & 0.995 \\
\midrule
1 & \multirow{3}{*}{Daily} & \multirow{3}{*}{0} & 0.001 & \multirow{3}{*}{0.995} \\
2 & & & 1 & \\
3 & & & 100 & \\
\bottomrule
\end{tabular}
\caption{Regularization parameters for forecast reconciliation. Main case updates parameters at every time step, whereas cases 1-3 only update parameters every 24 hours to avoid reusing observations and base forecasts for estimation.}
\label{tab:rrr_combined_config}
\end{table}

% Reconciliation
\paragraph{Forecast reconciliation.}
Forecasts were reconciled in the temporal hierarchy for each region separately. The same configuration was used for each region, with details provided in Table \ref{tab:rrr_combined_config}. Note in particular that the shrinkage target was chosen as zero, with varying levels of regularization being tested. In all cases, the regularization parameter $Q$ was chosen as the identity scaled by the value in the table. Two different setups were tested when reconciling; the first performs parameter updates at every data update (every hour), whilst the second only does so every 24 hours to avoid reusing base forecasts and observations. The latter setup generally yielded less accurate results with higher variance estimates (see Figure \ref{fig:mse_vs_avg_variance_east} for the east region). This is not surprising, as the parameter estimates are initially based on fewer observations, and may not adapt to changes as quickly as the former setup. In an attempt to mitigate the effect, we further investigate the use of hourly updates with different levels of regularization. In the following, we refer to the first setup as the main case, and the second setup as cases 1-3 depending on the level of regularization.
\begin{figure}
    \centering
    \includegraphics{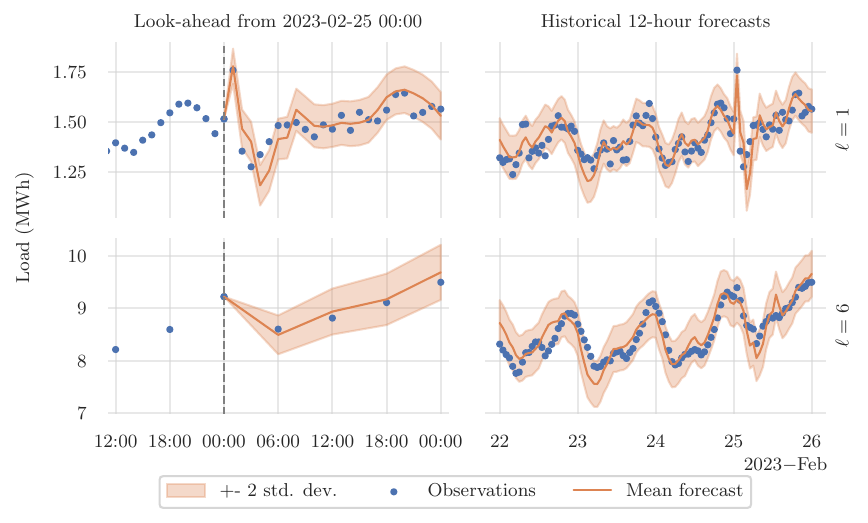}
    \caption{Sample of reconciled forecasts for the east region including prediction bands as +/- 2 standard deviations. For reference, observed values and aggregations of these are shown in blue. Left column shows forecasts for all horizons assuming information available up to the "current" hour, indicated by the vertical dashed line. The right column shows historical forecasts for a fixed 12-hour horizon. Top rows contain forecasts for the bottom level (hourly), while the bottom rows are for the 6-hour aggregation level.}
    \label{fig:forecast_samples_pred_int}
\end{figure}
Beginning with the main case, Figure \ref{fig:forecast_samples_pred_int} shows examples of the reconciled forecasts for the east region. In general, residuals were larger for the south region, which may be attributed to generally higher levels of heat consumption and variance. Comparing east and west, consumption levels are similar indicating that some difference in performance between the regions may be due to model design and hyperparameter tuning being performed specifically for the east region. Generally, errors increase with aggregation level and forecast horizon. To measure the effect of forecast reconciliation, we define the relative reduction in \gls{rmse} as,
\begin{equation*}
    \text{RRMSE} = \frac{\text{RMSE of base forecasts} - \text{RMSE of reconciled forecasts}}{\text{RMSE of base forecasts}}.
\end{equation*}
Forecasts reconciled using parameter updates every hour are improved for every region, level and horizon. Figures \ref{fig:sunburst_rrmse} and \ref{fig:rrmse_boxplot} summarize all \gls{rrmse} scores in the hierarchy. 
\begin{figure}
    \centering
    \includegraphics{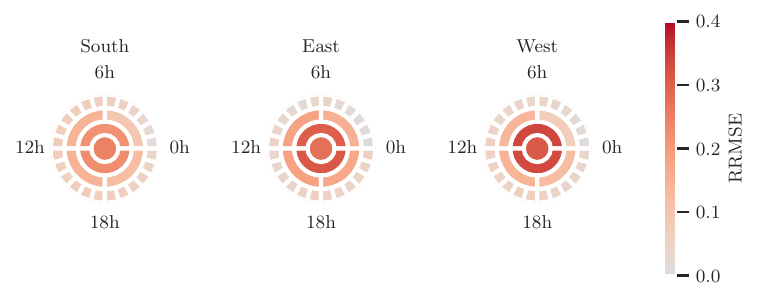}
    \caption{Relative reduction in \gls{rmse} achieved through forecast reconciliation (main case). Tiles represent forecasts organized such that increasing angle and decreasing radius match horizon and level.}
    \label{fig:sunburst_rrmse}
\end{figure}
\begin{figure}
    \centering
    \includegraphics{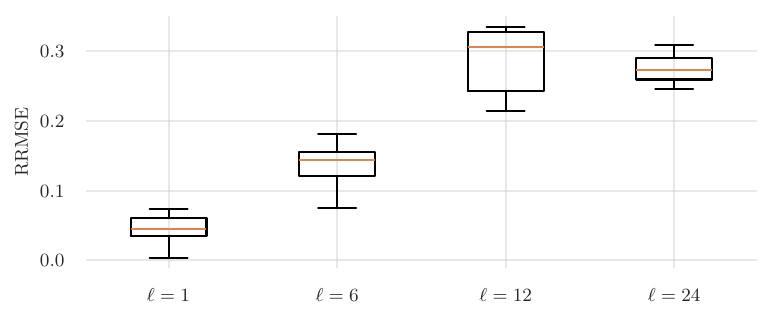}
    \caption{\gls{rrmse} of reconciled forecasts by aggregation level (main case).}
    \label{fig:rrmse_boxplot}
\end{figure}
% Regularization and parameter update frequency
\paragraph{Update frequency and regularization.}
When reconciliation in the temporal hierarchy is performed with masked data updates, mean and variance parameters are only updated once every 24 hours. Despite concerns on violating independence assumptions, updating parameters at every step generally yielded the best predictive performance. Figure \ref{fig:mse_vs_avg_variance_east} compares the \gls{mse} and average prediction error variance estimates for the east region for both approaches using the main case and case number 1. In the main case \gls{mse} and variance estimates follow a similar trend, increasing with horizon at a diminishing rate. For case 1, variance estimates become less stable and may overshoot realized errors significantly. Similar trends are observed in the other regions (see Figure \ref{fig:mse_vs_avg_variance_other_regions}).

For masked updates, cases 1-3 demonstrate the effects of varying regularization. Figure \ref{fig:rrmse_lineplot_all_regions_level_1} shows \gls{rrmse} for the bottom level reconciled forecasts in each region (see Figure \ref{fig:rrmse_diff_sunburst_all_regions_regularisation} in Appendix \ref{sec:appendix_figures} for a summary for all levels). As the extent of regularization increases, the \gls{rrmse} tends to increase for small horizons whilst it decreases for longer horizons. The base forecasts are updated at every step and use relatively short memory factors for the lower aggregation levels, implying that parameter estimates may change quickly. If important dynamics are captured by the parameter update scheme, reconciled forecast performance may degrade as parameter update frequency is reduced. Due to the choice of memory parameters, this effect should be most pronounced for the bottom level forecasts. Short horizons may also be more sensitive to fast changes, as they are more likely to reflect the current state of the system. Increasing the degree of regularization pushes the reconciled forecasts towards the bottom level base forecasts, thus providing better performance for short horizons. Conversely, for the longer horizons, the benefit of reconciliation is most pronounced when little regularization is applied.

Alternatively, the difference in how regularization affects forecasts in the hierarchy may be explained by worse top level base forecasts. Again, high regularization would push the reconciled forecasts towards the superior bottom level base forecasts. Increasing regularization for reconciliation when using hourly parameter updates showed no clear benefit in \gls{rrmse}, making the former explanation seem more plausible. 

\begin{figure}
    \centering
    \includegraphics{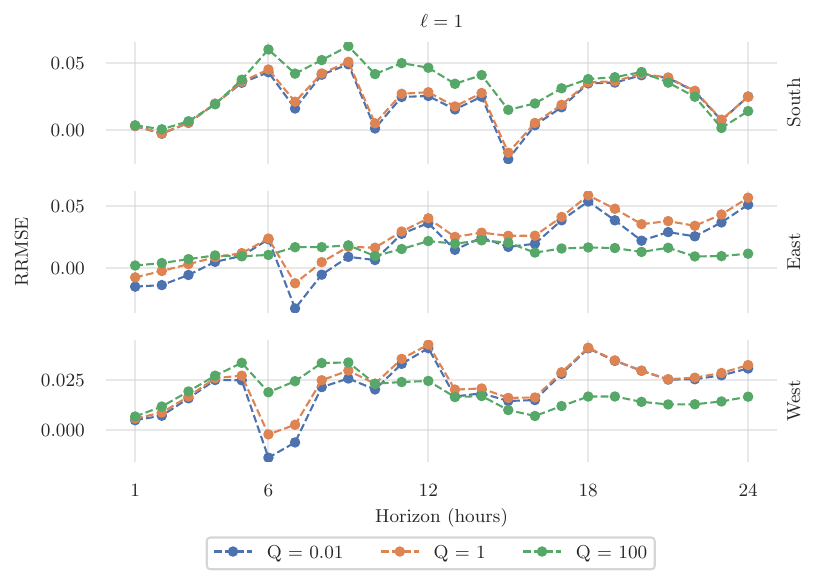}
    \caption{\gls{rrmse} scores for cases 1-3 for bottom level reconciled forecasts compared to base forecasts for each region and forecast horizon.}
    \label{fig:rrmse_lineplot_all_regions_level_1}
\end{figure}
\begin{figure}
    \centering
    \includegraphics{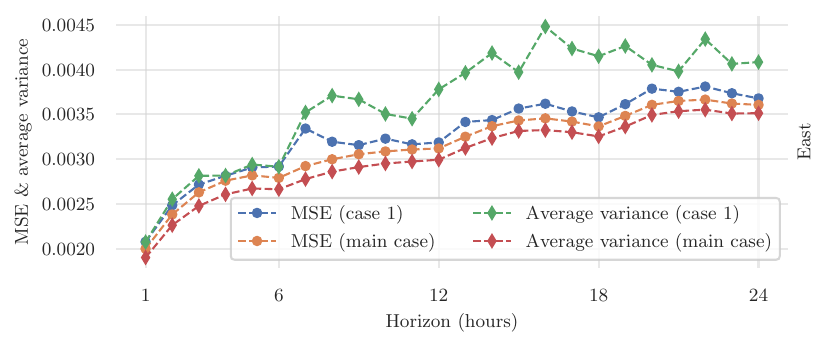} 
    \caption{MSE and averaged forecasts of prediction error variance estimates for the east region, based on reconciled forecasts using parameter updates every hour (main case) or every 24 hours (case 1).}
    \label{fig:mse_vs_avg_variance_east}
\end{figure}

\section{Conclusion}\label{sec:conclusion}
Formulating linear forecast reconciliation as a \gls{glm} enables the use of common statistical methods and algorithms suitable for online and recursive estimation and prediction. Based on the work of \cite{mollerOptimalForecastReconciliation2024}, a generic framework for representing linear hierarchical reconciliation was presented. This framework uses a suitable regularized multivariate linear model, enabling online adaptive learning and uncertainty quantification of parameters and forecasts. The framework begins with the \gls{glm} \eqref{eq:glm} and may use either of two approaches for inference, ridge regression or Bayesian estimation with matrix normal priors. The two approaches are closely related when regularization parameters are chosen as precision matrices of the row covariances of the residuals and the parameter priors. In this case, both approaches produce the same algebraic form of the parameter estimates and consequently the same point predictions of new data. However, due to different model assumptions, the two approaches differ in their uncertainty quantification. Depending on the application, either approach may be preferred. The ridge regression approach allows independent choices of $P$ and $U$, enabling a simple approach for implementing estimation schemes with weighting per observation by choosing $P$ accordingly. This is generally not possible in the Bayesian setting, where the interpretation $P=U^{-1}$ imposes modelling restrictions. To align with existing methods on \gls{rls}, the ridge regression approach was developed further by choosing the regularization parameter $P$ to define exponentially decaying weights for the observations. With this choice of $P$ and diagonal $U$, it was possible to derive simple recursive updates for parameter and prediction distributions. 

A more general treatment of the regularization matrix $P$ and structure of the covariance matrix $U$ for recursive estimation may be considered. Both matrices have row dimension equal to the number of observations used for estimation, which may eventually become problematically large in an online setting. To enable recursive updates, it may be useful to consider banded or sparse structures that discard long term correlations in the residuals. Similarly, for the column covariance, specific parametrizations as demonstrated by \cite{mollerLikelihoodbasedInferenceTemporal2024} may be beneficial. 

Whilst the \gls{glm} formalism used in this work may be specialized to usual shrinkage estimators used in forecast reconciliation, it also suggests the use of generic regularization techniques. Lasso \parencite{tibshiraniRegressionShrinkageSelection1996} or elastic net \parencite{zouRegularizationVariableSelection2005} regression methods are natural candidates. Both methods tend to produce sparse solutions, which are particularly useful for reducing the number of parameters in large hierarchies. In general however, the use of such methods does not permit closed form solutions, which are desirable for recursive estimation methods. Parameter selection strategies may also be performed based on statistical tests, as suggested by \cite{mollerOptimalForecastReconciliation2024}. In this case, the graph formulation of hierarchies may be useful, as it introduces an (informal) notion of distance between nodes using path lengths. It seems reasonable that more distant nodes are less likely to provide useful information for reconciling forecasts, suggesting that the corresponding weights be tested for removal.

For variance estimation, further use of the results from \cite{mollerOptimalForecastReconciliation2024} may also be incorporated. In particular, including a prior for the column, and potentially row, covariances may improve estimation in cases with limited data. The predictive distribution should then be generalized to account for the uncertainty of the variance estimates.

For general hierarchies, especially spatio-temporal hierarchies, further generalizations of the matrix variate linear model may be considered. In fact, much work on multiway or array data analysis already exists, see for example \cite{lockTensoronTensorRegression2018,manceurMaximumLikelihoodEstimation2013, nzabanita2015maximum}. If tensor normal distributions are assumed, it is known that the covariance structure may be parameterized by a sequence of covariance matrices, one for each node, joined by a Kronecker structure as in the matrix normal case (see e.g. \cite{manceurMaximumLikelihoodEstimation2013}). For spatio-temporal hierarchies, the separation of covariance structures may naturally decompose into spatial and temporal components as is characteristic of multiway data.

The recursive ridge regression method described in section \ref{sec:recursive_estimation} was implemented numerically and applied to a district heating load forecasting problem. This implementation uses the \textit{PyOnlineForecast} package, which uses a \gls{rls}-like algorithm based on the content of section \ref{sec:recursive_estimation}. The case used a 24-hour temporal hierarchy, and demonstrated that meaningful improvements in forecast accuracy can be achieved, in particular for longer horizons and higher aggregation levels. Regularization was applied to shrink reconciliation weights towards zero. As a consequence of using an online and adaptive parameter estimation scheme, forecast reconciliation was most successful when updates to parameter estimates were made at the frequency of the lowest temporal aggregation levels. Using this approach, all non-top level base forecast and all bottom level observations were used multiple times in parameter updates. When updates were instead restricted to the frequency of the fully aggregate 24-hour level, results were generally worse and required higher regularization to achieve the best performance. This effect is believed to be exclusive to online reconciliation in temporal hierarchies, as not all observations in the hierarchy are available at the same time. Relaxing either assumption, by considering non-temporal hierarchies or in sample performance, should alleviate this issue. In all cases, shrinkage regularization was necessary to avoid rank deficiency in the least squares problems for parameter estimation. Increasing shrinkage pushed results towards bottom-up reconciliation, which was required when updates were made less frequently, otherwise bottom level reconciled forecasts would lose accuracy compared to base forecasts.
\section*{Acknowledgement}
This work is supported by \textit{INSIEME} (Digital Europe No. 101194952),  \textit{ELEXIA} (Horizon Europe No. 101075656), \textit{ARV} (EU H2020 No. 101036723),  \textit{SEEDS} (Horizon Europe No. 10042024), \textit{BIPED} (Horizon Europe No. 101139060), \textit{Weather2X} (EUDP 640232-511303) and \textit{DynFlex}, which is one of the Innomission Green Fuel projects.

\printglossaries

\printbibliography
\appendix
\section{Useful identities}\label{sec:appendix_useful_identities}
Based on \cite{glanzExpectationMaximizationAlgorithmMatrix2013}, we say that a random variable $X$ taking values in $\mathbb{R}^{t \times m}$ follows a matrix normal distribution, written
\begin{equation*}
    X \sim \mathcal{MN}(M, U, V),
\end{equation*}
if its vectorization follows a multivariate normal distribution as,
\begin{equation*}
    \vect{X} \sim N(\vect{M^\top}, V \otimes U).
\end{equation*}
The density of $X$ is then given by,
\begin{equation*}
    (2\pi) ^{-\frac{1}{2} tm}|\left(V \otimes U\right)^{-1}|^\frac{1}{2} \exp \left( -\frac{1}{2} \text{tr}\left( V^{-1} (X - M)^\top U^{-1} (X - M) \right) \right).
\end{equation*}
Using \cite{petersen2008matrix} as a reference, we first note some useful identities. For matrices $X,Y, F, G$ with $F, G$ symmetric, the cyclic property of the trace easily yields,
\begin{equation}
    \text{tr} \left( G (X - Y)^\top F (X - Y) \right) = \text{tr} \left( G X^\top F X \right) + \text{tr} \left( G Y^\top F Y \right) - 2 \text{tr} \left( G X^\top F Y \right).
    \label{eq:trace_identity}
\end{equation}
The product of two matrix normal densities with the same column covariance matrix $V$ is also matrix normal. Letting $\Phi(\cdot \mid \mathbb{M}, \mathbb{U}, V)$ denote the density of $\mathcal{MN}(\mathbb{M}, \mathbb{U}, V)$ on $\mathbb{R}^{t \times m}$, then the following identity holds,
\begin{equation*}
    \Phi(\cdot \mid M_1, U_1, V) \Phi(\cdot \mid M_2, U_2, V) \propto \Phi(\cdot \mid M, U, V),
\end{equation*}
where the row covariance and mean parameters are given by,
\begin{equation*}
    U = V \otimes (U_1^{-1} + U_2^{-1})^{-1}, \quad M = (U_1^{-1} + U_2^{-1})^{-1} (U_1^{-1} M_1 + U_2^{-1} M_2).
\end{equation*}
The result follows from vectorizing and using properties of the multivariate normal distribution. To show this, we use the well known result, that the product of the densities of two independent multivariate normal distributions $\mathcal{N}(\mu_1, \Sigma_1)$ and $\mathcal{N}(\mu_2, \Sigma_2)$, is proportional to another multivariate normal density $\mathcal{N}(\mu, \Sigma)$ with,
\begin{equation*}
    \Sigma = (\Sigma_1^{-1} + \Sigma_2^{-1})^{-1}, \quad \mu = \Sigma (\Sigma_1^{-1} \mu_1 + \Sigma_2^{-1} \mu_2).
\end{equation*}
Substituting the parameters of the vectorized forms, $\vect{M_1}$, $\vect{M_2}$, $\Sigma_1 = V \otimes U_1$ and $\Sigma_2 = V \otimes U_2$ yields,
\begin{equation*}
    \Sigma = (V^{-1} \otimes U_1^{-1} + V^{-1} \otimes U_2^{-1})^{-1} = V \otimes (U_1^{-1} + U_2^{-1})^{-1},
\end{equation*}
and using the mixed product property of the Kronecker product,
\begin{equation*}
    \begin{aligned}
        \mu &= \Sigma ( (V^{-1} \otimes U_1^{-1}) \vect{M_1} + (V^{-1} \otimes U_2^{-1}) \vect{M_2} ) \\
        &= (V \otimes (U_1^{-1} + U_2^{-1})^{-1}) \vect{ (U_1^{-1} M_1 + U_2^{-1} M_2) V^{-1} } \\
        &= \vect{ (U_1^{-1} + U_2^{-1})^{-1} (U_1^{-1} M_1 + U_2^{-1} M_2) }.
    \end{aligned}
\end{equation*}
From the vectorized form of the matrix normal distribution, it is clear that $\mu, \Sigma$ is simply the mean and covariance of $\vect{Z}$, so the result follows. A similar result holds for the sum of two independent matrix normally distributed variables with the same column covariance matrix $V$. Let,
\begin{equation*}
    Z_1 \sim \mathcal{MN}(M_1, U_1, V), \quad Z_2 \sim \mathcal{MN}(M_2, U_2, V),
\end{equation*}
then,
\begin{equation}
    Z_1 + Z_2 \sim \mathcal{MN}(M_1 + M_2, U_1 + U_2, V).
    \label{eq:matrix_normal_independent_sums}
\end{equation}
This follows directly from vectorizing and using properties of the multivariate normal distribution and reducing the sum of the variances as,
\begin{equation*}
    \var{\vect{Z_1 + Z_2}} = \var{\vect{Z_1}} + \var{\vect{Z_2}} = V \otimes U_1 + V \otimes U_2 = V \otimes (U_1 + U_2).
\end{equation*}
Finally, we note that for $Z \sim \mathcal{MN}(M, U, V)$, it holds that for a constant matrix $A$,
\begin{equation}
    A Z \sim \mathcal{MN}(A M, A U A^\top, V).
    \label{eq:matrix_normal_linear_transform}
\end{equation}
This follows by vectorizing to $\vect{Z}$, scaling the variance by $(I \otimes A)$, and reducing the vectorized covariance given by,
\begin{equation*}
    \var{\vect{A Z}} = (I \otimes A) \var{\vect{Z}} (I \otimes A)^\top = (I \otimes A) (V \otimes U) (I \otimes A^\top) = V \otimes (A U A^\top).
\end{equation*}
\section{Bayesian regression}\label{sec:appendix_bayesian_regression}
The usual linear regression model with normal errors and parameter prior, provides well known results for the posterior and predictive distributions. For the sake of completeness, we provide a derivation of corresponding results for the matrix variate linear model with matrix normal errors and priors. In the Bayesian setup, we assume $U$, $V$, $\Psi$ and $\Omega$ known and \gls{spd}, and start from the model,
\begin{equation}
    Y \mid \theta \sim \mathcal{MN}(X \theta, U, V), \quad \theta \sim \mathcal{MN}(\theta_0, \Psi_0, \Omega_0).
    \label{eq:bayesian_model}
\end{equation}
This is the same as the model in \eqref{eq:glm}, with the addition of a matrix normal prior on the parameters with mean $\theta_0$ and covariance parameters $\Psi_0$ and $\Omega_0$.
\subsection{Posterior and MAP estimate}
The linear model \eqref{eq:bayesian_model} lends itself well to Bayesian inference due to its simple linear structure, and conjugate matrix normal prior.
\begin{proposition}[Posterior distribution]\label{prop:bayesian_posterior}
Under the model \eqref{eq:bayesian_model}, if $V = \Omega_0 = \Sigma$, then the posterior distribution of $\theta$ is given by,
\begin{equation}
    \theta | Y \sim \mathcal{MN} (\hat{\theta}, \Psi, \Sigma),
    \label{eq:posterior_matrix_normal}
\end{equation}
with,
\begin{equation*}
    \hat{\theta} = \Psi \left( X^\top U^{-1} Y + \Psi_0^{-1} \theta_0 \right), \quad \Psi = \left( X^\top U^{-1} X + \Psi_0^{-1} \right)^{-1}.
\end{equation*}
Since $\hat{\theta}$ maximizes the Gaussian posterior it is the \gls{map} estimate.
\end{proposition}
\begin{proof}
In general, it holds for the posterior distribution of $\theta$ that,
\begin{equation*}
    p(\theta | Y) \propto p(Y | \theta) p(\theta).
\end{equation*}
To show the result, we argue that the \gls{mle} of $\theta$ can be expressed as a matrix normal distribution,
\begin{equation}
    p(Y | \theta) \propto \mathcal{MN} \left( \theta; \hat{\theta}_\text{l}, \Psi_\text{l}, V \right).
    \label{eq:likelihood_as_matrix_normal}
\end{equation}
Where the \gls{mle} parameters are given by,
\begin{equation}
    \hat{\theta}_\text{l} = \left( X^\top U^{-1} X \right)^{-1} X^\top U^{-1} Y, \quad \Psi_\text{l} = \left( X^\top U^{-1} X \right)^{-1}.
    \label{eq:mle_parameters}
\end{equation}
Let $\Phi(\cdot \mid \mathbb{M}, \mathbb{U}, \mathbb{V})$ denote the density of the matrix normal distribution with mean $\mathbb{M}$ and covariance parameters $\mathbb{U}$ and $\mathbb{V}$. Then the posterior is proportional to the product of two matrix normal distributions,
\begin{equation}
    p(\theta | Y) \propto \Phi(\theta \mid \hat{\theta}_\text{l}, \Psi_\text{l}, V) \Phi(\theta \mid \theta_0, \Psi_0, \Omega).
    \label{eq:posterior_product_of_matrix_normals}
\end{equation}
We will show that for $V = \Omega = \Sigma$ this product is exactly \eqref{eq:posterior_matrix_normal}. In particular, for some scalar constant $C$,
\begin{equation}
    p(\theta | Y) = C \Phi(\theta \mid \hat{\theta}, \Psi, \Sigma).
    \label{eq:posterior_as_matrix_normal}
\end{equation}
but for this to be a valid density, we must have $C = 1$, and we have found the posterior distribution. To prove the statements used above, we rely on the assumption that $U$, $V$, $\Psi$ and $\Omega$ are \gls{spd}, and use the properties in Appendix \ref{sec:appendix_useful_identities}. To show \eqref{eq:likelihood_as_matrix_normal}, note that the likelihood can be written as,
\begin{equation*}
    p(Y | \theta) \propto \exp \left( -\frac{1}{2} \norm{Y - X \theta}_{U^{-1}, V^{-1}}^2 \right).
\end{equation*}
Then using \eqref{eq:trace_identity}, we rewrite the exponent as a quadratic form in $\theta$,
\begin{equation*}
    \begin{aligned}
        \norm{Y - X \theta}_{U^{-1}, V^{-1}}^2 &= \text{tr} \left(V^{-1} Y^\top U^{-1} Y \right) && \text{constant terms} \\
        &- 2 \text{tr} \left(V^{-1} Y^\top U^{-1} X \theta \right) && \text{1st order terms} \\
        &+ \text{tr} \left(V^{-1} \theta^\top X^\top U^{-1} X \theta \right) && \text{2nd order terms}.
    \end{aligned}
\end{equation*}
For \eqref{eq:likelihood_as_matrix_normal} to be true, the first and second order terms must match those of the exponent of the matrix normal distribution, while the constant terms can be absorbed into the proportionality constant. Choosing the parameters as in \eqref{eq:mle_parameters}, the exponent for the matrix normal distribution \eqref{eq:likelihood_as_matrix_normal} becomes,
\begin{equation*}
    \begin{aligned}
    \norm{\theta - \hat{\theta}_\text{l}}_{\Psi_\text{l}^{-1}, V^{-1}}^2 &= \text{tr} \left( V^{-1} \theta^\top \Psi_\text{l}^{-1} \theta \right)  - 2 \text{tr} \left( V^{-1} \theta^\top \Psi_\text{l}^{-1} \hat{\theta}_\text{l} \right) + \text{tr} \left( V^{-1} \hat{\theta}_\text{l}^\top \Psi_\text{l}^{-1} \hat{\theta}_\text{l} \right). \\
    & = \text{tr} \left( V^{-1} \theta^\top X^\top U^{-1} X \theta \right) - 2 \text{tr} \left( V^{-1} \theta^\top X^\top U^{-1} Y \right) + \text{constant terms}.
    \end{aligned}
\end{equation*}
Up to a constant, this matches the first and second order terms of the likelihood expression, so \eqref{eq:likelihood_as_matrix_normal} holds. To show that \eqref{eq:posterior_product_of_matrix_normals} is identical to \eqref{eq:posterior_as_matrix_normal}, we use the result from Appendix \ref{sec:appendix_useful_identities}. We assume $V = \Omega = \Sigma$, then,
\begin{equation*}
    \begin{aligned}
        \Psi &= \Sigma \otimes (\Psi_\text{l}^{-1} + \Psi_0^{-1})^{-1} = \left( X^\top U^{-1} X + \Psi_0^{-1} \right)^{-1}, \\
        \hat{\theta} &= (\Psi_\text{l}^{-1} + \Psi_0^{-1})^{-1} (\Psi_\text{l}^{-1} \hat{\theta}_\text{l} + \Psi_0^{-1} \theta_0) \\
        &= \left( X^\top U^{-1} X + \Psi_0^{-1} \right)^{-1} \left( X^\top U^{-1} Y + \Psi_0^{-1} \theta_0 \right).
    \end{aligned}
\end{equation*}
This shows that \eqref{eq:posterior_product_of_matrix_normals} is indeed equal to \eqref{eq:posterior_as_matrix_normal} with $C = 1$, and that the posterior distribution is given by \eqref{eq:posterior_matrix_normal}.
\end{proof}
% Predictive distribution
\subsection{Predictive distribution}
For a new pair of observations $Y^*$ and features $X^*$, we assume,
\begin{equation*}
    Y^* \mid \theta \sim \mathcal{MN} \left( X^* \theta, U^*, \Sigma \right).
\end{equation*}
We assume that conditional on $\theta$, $Y^*$ is independent of $Y$, so that
\begin{equation*}
    \left( X^* \theta + E^* \mid Y \right) = X^* \theta \mid Y + E^*,
\end{equation*}
where $E^*$ is a noise term,
\begin{equation*}
    E^* \sim \mathcal{MN} \left( 0, U^*, \Sigma \right).
\end{equation*}
Using equation \eqref{eq:matrix_normal_linear_transform}, we have that,
\begin{equation*}
    X^* \theta \mid Y \sim \mathcal{MN} \left(X^* \hat{\theta}, X^* \Psi X^{* \top}, \Sigma \right),
\end{equation*}
Since both distributions are matrix normal and share the same column covariance $\Sigma$, by equation \eqref{eq:matrix_normal_independent_sums}, the sum is also matrix normal,
\begin{equation}
    Y^* \mid Y \sim \mathcal{MN} \left( X^* \hat{\theta}, X^* \Psi X^{* \top} + U^*, \Sigma \right).
    \label{eq:predictive_distribution}
\end{equation}
% Frequentist parameter variance
By a similar argument as above, using instead the parameters from \eqref{eq:mle_parameters}, we get the \gls{mle} predictive distribution with known variance parameters,
\begin{equation*}
    Y^*_\text{l} \sim \mathcal{MN} \left( X^* \hat{\theta}_\text{l}, X^* \Psi_\text{l} X^{* \top} + U^*, \Sigma \right).
\end{equation*} 

\section{Ridge regression}\label{sec:appendix_ridge_regression}
% Ridge regression
If we modify the Bayesian setup by treating the precision matrices $U^{-1}$, $\Psi_0^{-1}$, $\Sigma^{-1}$ as regularization parameters $P$, $Q$, $R$ and do not explicitly model a prior on $\theta$, the problem becomes equivalent to the ridge regression problem posed in section \ref{sec:GLM}. In particular, using model \eqref{eq:glm} and minimizing the ridge regression objective \eqref{eq:ridge_regression_objective} we get the same point estimate for the parameters, i.e.,
\begin{equation}
    \hat{\theta} = \left(K + Q\right)^{-1} \left( L + Q \theta_0 \right), \quad K = X^\top P X, \quad L = X^\top P Y.
    \label{eq:ridge_regression_solution_appendix}
\end{equation}
This is true, since the maximizer of the normal posterior is the minimizer of the the loss function \eqref{eq:ridge_regression_objective} when identifying $P = U^{-1}$, $Q = \Psi_0^{-1}$ and $R = \Sigma^{-1}$. In particular, the minimizer of the loss, is the mean of the corresponding Gaussian posterior. Note, in the Bayesian setting, we condition on $Y$ so that $\hat{\theta}$ is a deterministic quantity. This is not true for the ridge regression model, and $\hat{\theta}$ should be considered a random variable that depends on the observations $Y$ and noise term $E$ as random variables. It is easy to see, that the only source of variance in $\hat{\theta}$ comes from the observation noise, i.e. the term,
\begin{equation*}
    \left(K + Q\right)^{-1} X^\top P E, \quad E \sim \mathcal{MN}(0, U, V).
\end{equation*}
We are thus dealign with a scaled matrix random variable, and using \eqref{eq:matrix_normal_linear_transform}, it follows that,
\begin{equation*}
    \hat{\theta} \sim \mathcal{MN} \left( \theta_\text{ridge}, \Psi_\text{ridge}, V \right).
\end{equation*}
Here the mean parameter is given by the expectation,
\begin{equation*}
    \theta_\text{ridge} = \expect{\hat{\theta}} = \left(K + Q\right)^{-1} \left(X^\top P X \theta + Q \theta_0 \right),
\end{equation*}
and the covariance parameter is given by,
\begin{equation*}
    \Psi_\text{ridge} = \left(K + Q\right)^{-1} H \left(K + Q\right)^{-1}, \quad H = X^\top P U P X.
\end{equation*}
Note that $\theta_\text{ridge}$ is biased, and is never evaluated in practice since the true parameter $\theta$ is unknown. In contrast, the covariance parameter $\Psi_\text{ridge}$ is a function of known quantities only and is useful for quantifying parameter and predictive uncertainty. The prediction model implied by \eqref{eq:glm} is,
\begin{equation*}
    Y^*_\text{ridge} = X^* \hat{\theta} + E^*, \quad E^* \sim \mathcal{MN}(0, U^*, V).
\end{equation*}
Using identities \eqref{eq:matrix_normal_linear_transform} and \eqref{eq:matrix_normal_independent_sums}, the predictive distribution with known variance parameters is then matrix normal with,
\begin{equation}
    Y^*_\text{ridge} \sim \mathcal{MN} \left( X^* \theta_\text{ridge}, X^* \Psi_\text{ridge} X^{* \top} + U^*, V \right).
    \label{eq:ridge_regression_predictive_distribution}
\end{equation}  

\subsection{Recursive updates}\label{sec:appendix_recursive}
For the linear model \eqref{eq:glm_t}, estimation becomes computationally expensive as $t$ increases. Recursive updates for the ridge regression parameters and predictions can be derived in some special cases. To begin, denote features and observations up to time step $t$ using subscript $t$, e.g. $X_t$ and $Y_t$. Similarly, let $K_t$, $L_t$ and $H_t$ correspond to the matrices $K$, $L$ and $H$ defined in the previous. When $P=P_t$ and $U=U_t$ are diagonal matrices with particularly well-behaved updates, we can derive simple closed forms for the posterior updates. If $P_t$ is chosen as an exponentially decaying weight matrix, i.e. with diagonal elements $p_i = \lambda^{t-i}$ as in corollary \ref{cor:recursive_updates}, we recover a form of \gls{rls} solutions. The proof is straightforward, see below.

\begin{proof}[Proof of corollary \ref{cor:recursive_updates}]
    Follows by evaluating the expressions for $K_{t+1}$, $L_{t+1}$ and $H_{t+1}$ as sums over samples. Let $P_t = \text{diag}(\lambda^{t-1}, \lambda^{t-2}, \ldots, \lambda^0)$. For $K_{t+1}$, we have,
\begin{equation*}
    \begin{aligned}
        K_{t+1} &= X_{t+1}^\top P_{t+1} X_{t+1} \\
        &= \sum_{i=1}^{t+1} \lambda^{t+1-i} x_i x_i^\top \\
        &= \lambda \sum_{i=1}^{t} \lambda^{t-i} x_i x_i^\top + x_{t+1} x_{t+1}^\top \\
        &= \lambda K_t + x_{t+1} x_{t+1}^\top,
    \end{aligned}\end{equation*}
    and similarly for $L_{t+1}$ in terms of $L_t$. Letting $u_i$, $i = 1, \ldots, t$ denote the diagonal elements of $U_t$, for $H_t$ we have,
\begin{equation*}
    \begin{aligned}
        H_t &= X_t^\top P_t U_t P_t X_t \\
        &= \sum_{i=1}^{t} \sum_{j=1}^{t} U_{t,ij} \lambda^{t-i} \lambda^{t-j} x_i x_j^\top \\
        &= \sum_{i=1}^{t} u_{t,i} \lambda^{2(t-i)} x_i x_i^\top.
    \end{aligned}
\end{equation*}
For $H_{t+1}$ we get,
\begin{equation*}
    \begin{aligned}
        H_{t+1} &= \sum_{i=1}^{t+1} u_{t+1,i} \lambda^{2(t+1-i)} x_i x_i^\top \\
        &= \lambda^2 \sum_{i=1}^{t} u_{t,i} \lambda^{2(t-i)} x_i x_i^\top + u_{t+1,t+1} x_{t+1} x_{t+1}^\top \\
        &= \lambda^2 H_t + u_{t+1,t+1} x_{t+1} x_{t+1}^\top.
    \end{aligned}
\end{equation*}
\end{proof}
This result allows us to recursively update the parameter estimate, mean and covariance. In practice, we are most interested in the parameter estimate and its covariance,
\begin{equation*}
    \begin{aligned}
        \hat{\theta}_{t+1} &= \left(K_{t+1} + Q\right)^{-1} \left( L_{t+1} + Q \theta_0 \right), \\
        \Psi_{\text{ridge}, t+1} &= \left(K_{t+1} + Q\right)^{-1} H_{t+1} \left(K_{t+1} + Q\right)^{-1}.
    \end{aligned}
\end{equation*}
The predictions are then $y_{t+1}^* = x^*_{t+1} \hat{\theta}_{t+1}$, for which the distribution \eqref{eq:ridge_regression_predictive_distribution} collapses to,
\begin{equation*}
    y_{t+1}^* \sim \mathcal{N} \left( \theta_{\text{ridge}, t+1} x^*_{t+1}, (x^{* \top}_{t+1} \Psi_{\text{ridge}, t+1} x^*_{t+1} + u^*_{t+1}) V \right). 
\end{equation*}
Note, equivalent recursive results could be stated for the Bayesian model, by updating the posterior parameters and the predictive distribution in a similar manner. However, when $P_t$ is considered a hyperparameter, the implied restriction on the variance $U_t = P_t^{-1}$ may be undesirable. In particular, when the diagonal elements $p_{t,i}$ of $P_t$ depend on $t$, the variance model for past observations may tacitly change between time steps.
%\newpage
%\input{sections/observed_basis.tex}
\newpage
\section{Base forecast models \& transformations}\label{sec:appendix_base_models}
To generate base forecasts, the \textit{PyOnlineForecast} package was used \parencite{Ronlev-Knudsen_PyOnlineForecast_2026}. The package provides multiple built-in transformations to be used for feature engineering. Table \ref{tab:transformations} provides a description of the transformations used for generating base forecasts. Input data names and abbreviations for use in the following tables are summarized in Table \ref{tab:data_variable_abbreviations}. For each region and node in the temporal hierarchy \ref{fig:hierarchy_24h}, a separate base forecast model was instantiated, and parameters estimated. This results in a total of 93 models to be instantiated. Each model uses multiple features, meaning that individual descriptions of models is not feasible. Instead, Table \ref{tab:input_features} summarizes all features used across all models, and Table \ref{tab:base_model_features} summarizes their use in forecast model ensembles. Each ensemble maintains models for each horizon for a single level in the hierarchy. The input features are selected to match the forecast horizon of the model, meaning that for horizon $k > 0$, only features with horizon $k$ (forecasts) or $0$ (current, usually shared values) are included.
\begin{table}
\centering
\small
\begin{tabularx}{\textwidth}{lX}
\toprule
\textbf{Transformation} & \textbf{Description} \\
\midrule
\texttt{GetItem} & Selects from output using built in \texttt{\_\_getitem\_\_} methods. \\
\texttt{Subset} & Selects subset of columns based on variable names and forecast horizons. \\
\texttt{DataCleaner} & Provides outlier detection and removal, forward filling, resampling. Uses a running estimate of mean and variance with exponential forgetting factor. \\
\texttt{Scaler} & Multiplies selected columns by a scaling factor. \\
\texttt{Aggregator} & Resamples time series using summation or averaging. \\
\texttt{Index} & Extracts the row index. \\
\texttt{Align} & Combines multiple \texttt{DataFrame}s by reindexing data a shared index and concatenating. \\
\texttt{One} & Provides constant value of one (for intercept terms). \\
\texttt{TimeOfDay} & Computes time of day as a fraction. \\
\texttt{TimeOfWeek} & Provides time of week features. \\
\texttt{FourierSeries} & Computes terms of a Fourier series. \\
\texttt{LowPass} & Applies exponential smoothing filter. \\
\texttt{Disruption} & Provides square wave disruption features. \\
\texttt{Lag} & Shifts time series by specified lag. \\
\texttt{SlidingSum} & Computes sliding window sums. \\
\bottomrule
\end{tabularx}
\caption{Built-in transformations provided by the \textit{PyOnlineForecast} package.}
\label{tab:transformations}
\end{table}
\begin{table}
\centering
\footnotesize
\begin{tabular}{ll}
\toprule
\textbf{Abbreviation} & \textbf{Full Variable Name} \\
\midrule
$E_{\text{east}}$ & \texttt{energi\_oest\_MW} \\
$E_{\text{south}}$ & \texttt{energi\_syd\_MW} \\
$E_{\text{west}}$ & \texttt{energi\_vest\_MW} \\
$T_{\text{local}}$ & \texttt{t\_ambient\_tranbjerg\_degC} \\
$T_{\text{air}}$ & \texttt{air\_temperature\_2m} \\
$RH$ & \texttt{relative\_humidity\_2m} \\
$WS$ & \texttt{wind\_speed\_10m} \\
$SW$ & \texttt{integral\_of\_surface\_downwelling\_shortwave\_flux\_in\_air\_wrt\_time} \\
\bottomrule
\end{tabular}
\caption{Data variable abbreviations.}
\label{tab:data_variable_abbreviations}
\end{table}
\begin{table}
\centering
\footnotesize
\begin{tabular}{l|l|l|l}
\toprule
\textbf{Abbr.} & \textbf{Type} & \textbf{Input(s)} & \textbf{Parameters} \\
\midrule
\multicolumn{4}{c}{\textit{Data Sources}} \\
\midrule
S1 & Source & -- & name="load" \\ % r_load
S2 & Source & -- & name="weather\_forecast" \\ % r_weather_forecast
S3 & Source & -- & name="local\_temp" \\ % r_local_temp
\midrule
\multicolumn{4}{c}{\textit{Pre-processing}} \\
\midrule
P1 & GetItem & S1 & $E_{\text{east}}$, $E_{\text{south}}$, $E_{\text{west}}$ \\ % sub_load
P2 & Subset & S2 & $T_{\text{air}}$, $RH$, $WS$, $SW$, horizons=[1..24] \\ % sub_weather_forecast
P3 & GetItem & S3 & $T_{\text{local}}$ \\ % sub_local_temp
P4 & DataCleaner & P1 & $\alpha$=0.995, freq="5min" \\ % clean_load
P5 & DataCleaner & P2 & $\alpha$=0.995, freq="1h" \\ % clean_weather_forecast
P6 & DataCleaner & P3 & $\alpha$=0.995, freq="1h" \\ % clean_local_temp
P7 & Scaler & P5 & $SW$: 1e-5 \\ % scaled_weather_forecast
P8 & Aggregator & P4 & freq="h", agg="mean" \\ % agg_load
P9 & Index & P7 & -- \\ % index
P10 & Align & P9, P8, P7, P6 \\ % clean_data
\midrule
\multicolumn{4}{c}{\textit{Features}} \\
\midrule
F1 & Index & P10 & -- \\ % clean_index
F2 & One & P10 & -- \\ % const
F3 & TimeOfDay & F1 & -- \\ % tod
F4 & TimeOfWeek & F1 & -- \\ % tow
F5 & FourierSeries & F3 & n=8 \\ % fs_day
F6 & FourierSeries & F4 & n=8 \\ % fs_week
F7 & GetItem & P10 & $T_{\text{air}}$ \\ % air_temp_2m
F8 & GetItem & P10 & $RH$ \\ % rel_humidity_2m
F9 & GetItem & P10 & $T_{\text{local}}$ \\ % local_temp
F10 & GetItem & P10 & $WS$ \\ % wind_speed_10m
F11 & LowPass & F7 & $\alpha$=0.95 \\ % lp_air_temp_lo
F12 & LowPass & F7 & $\alpha$=0.2 \\ % lp_air_temp_hi
F13 & LowPass & F8 & $\alpha$=0.9 \\ % lp_hum
F14 & LowPass & F9 & $\alpha$=0.99 \\ % lp_t_tran
F15 & LowPass & F10 & $\alpha$=0.8 \\ % lp_wind
F16 & GetItem & P10 & $T_{\text{local}}$, $WS$, $T_{\text{air}}$, $SW$ \\ % shared_vars
F17 & Disruption$^\dagger$ & F1 & hour=[1..7], dayofweek=5, horizons=[1..24] \\ % disruptions
\midrule
\multicolumn{4}{c}{\textit{Energy Features$^{*}$}} \\
\midrule
E1 & GetItem & P10 & \texttt{energi\_\{r\}\_MW} ($r \in \{\text{syd}, \text{oest}, \text{vest}\}$) \\ % energy
E2 & LowPass & E1 & $\alpha$=0.2 \\ % lp_energy
E3 & LowPass & E1 & $\alpha$=0.9 \\ % lp_energy
E4 & Lag & E3 & lag=24, offsets=[1..24] \\ % ar_lp_24
E5 & Lag & E3 & lag=48, offsets=[12, 24] \\ % ar_lp_48
E6 & Lag & E3 & lag=72, offsets=[12, 24] \\ % ar_lp_72
E7 & Lag & E3 & lag=168, offsets=[12, 24] \\ % ar_lp_168
\midrule
\multicolumn{4}{c}{\textit{Targets$^{*}$}} \\
\midrule
Y1 & (reference) & E1 & -- \\ % Y[1]
Y2 & SlidingSum & E1 & window=6 \\ % Y[6]
Y3 & SlidingSum & E1 & window=12 \\ % Y[12]
Y4 & SlidingSum & E1 & window=24 \\ % Y[24]
\bottomrule
\end{tabular}
\caption{Overview of the base forecast model features and target variables. $^{*}$Region-specific (South, east and west), $^\dagger$ seven separate transformations.}
\label{tab:input_features}
\end{table}

\begin{table}
\centering
\small
\begin{tabularx}{\textwidth}{l|X|l|l|l}
\toprule
\textbf{Level} & \textbf{Features (X)} & \textbf{Target (Y)} & \textbf{Horizons} & \textbf{Memory} \\
\midrule
1 & F2, F5, F12, F11, F13, F16, E1, E2, E3, E4, F17$^{\dagger}$ & Y1 & 1, 2, ..., 24 & 0.995 \\ 
6 & F2, F5, F12, F11, F13, F16, E1, E2, E3, E4, F17$^{\dagger}$ & Y2 & 6, 12, 18, 24 & 0.995  \\
12 & F2, F6, F11, F13, F16, E3, E4, E5, E6, E7, F17$^{\dagger}$ & Y3 & 12, 24 & 0.995  \\
24 & F2, F6, F11, F14, F15, F13, F16, E3, E5, E6, E7, F17$^{\dagger}$ & Y4 & 24 & 0.9999  \\
\bottomrule
\end{tabularx}
\caption{Forecast model definitions using \texttt{ForecastEnsemble}. For each level, features and target variable are used to construct transformations that generate forecasts for the specified target and horizons. All regions use regularization with \texttt{Q}=0.001, \texttt{theta0} = 0 and \texttt{track\_memory} enabled for covariance estimation. For horizon $h=k$, only features with horizon $k$ (forecasts) or $0$ (observations) are included. $^{\dagger}$F17 (disruptions) only included in the east region model.}
\label{tab:base_model_features}
\end{table}

\newpage
\section{Additional figures}\label{sec:appendix_figures}
Figure \ref{fig:forecast_samples} shows a comparison of base and reconciled forecast for the east region in the main case.
\begin{figure}
    \centering
    \includegraphics{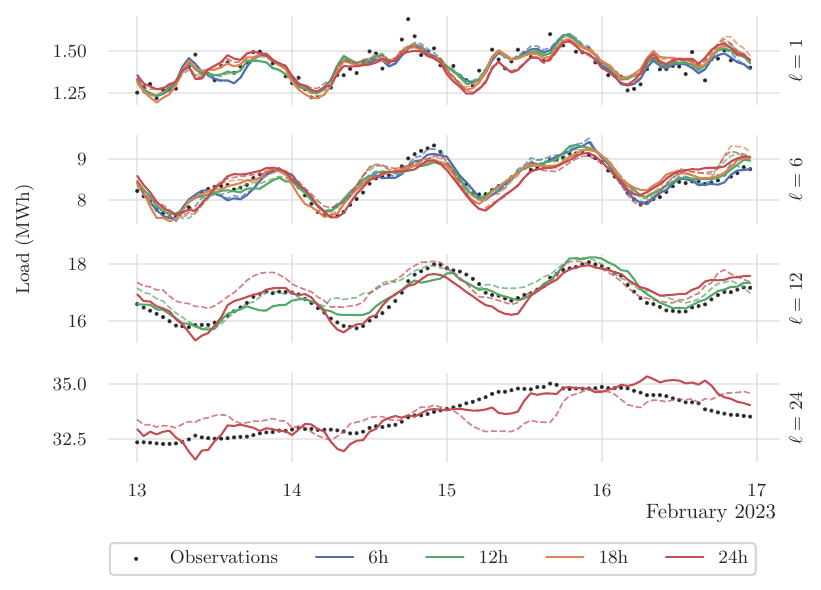}
    \caption{Sample of forecasts for the east region and all aggregation levels. Base forecasts are shown in dashed, whilst reconciled forecasts are shown in solid. Only forecasts for horizons 6, 12, 18 and 24 hours are shown.}
    \label{fig:forecast_samples}
\end{figure}
Figures \ref{fig:sunburst_rmse} and \ref{fig:sunburst_r_squared} provide additional plots of the results for the main case, while Figure \ref{fig:rrmse_diff_sunburst_all_regions_regularisation} summarizes the \gls{rrmse} for cases 1-3. For all figures, tiles represent forecasts organized such that increasing angle and decreasing radius match horizon and level. We define the $R^2$ score as the fraction of variance explained,
\begin{equation*}
    R^2 = 1 - \frac{\sum_{t=1}^T (y_t - \tilde{y}_t)^2}{\sum_{t=1}^T (y_t - \bar{y})^2}.
\end{equation*}
$R^2$ scores generally decrease with increasing forecast horizon, but increase with aggregation level. Both south and west regions show lower $R^2$ scores compared to east, which may again be attributed to model design and hyperparameter tuning being targeted the east region. \gls{rrmse} is generally higher for longer horizons and higher aggregation levels.
\begin{figure}
    \centering
    \includegraphics{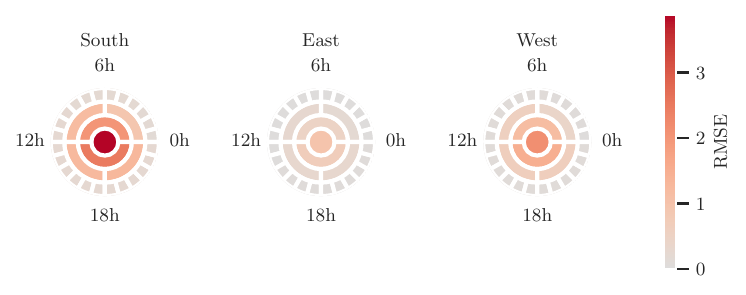}
    \caption{\gls{rmse} score for base forecasts of each region.}
    \label{fig:sunburst_rmse}
\end{figure}
\begin{figure}
    \centering
    \includegraphics{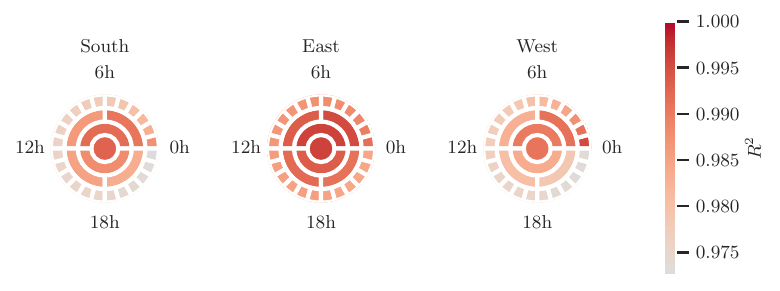}
    \caption{Fraction of variance explained by reconciled forecasts ($R^2$) from the main case.}
    \label{fig:sunburst_r_squared}
\end{figure}
\begin{figure}
    \centering
    \includegraphics{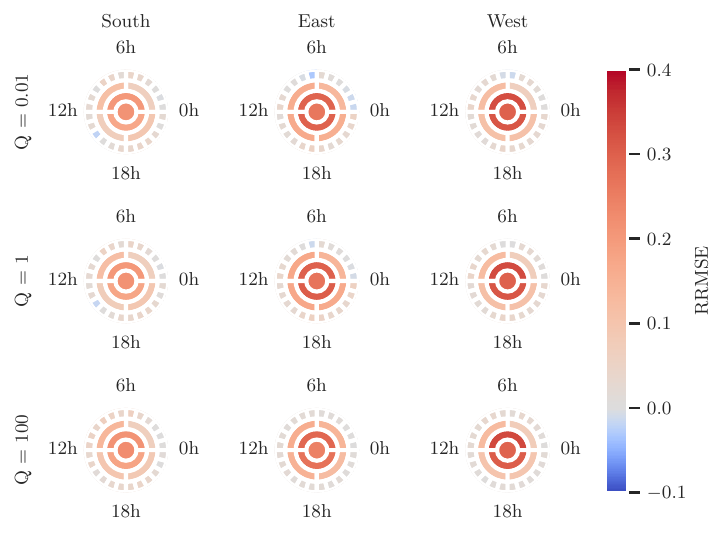}
    \caption{\gls{rrmse} for increasing regularization for each region, aggregation level, and forecast horizon for cases 1-3. Note that the colour scale is asymmetrical around zero, since all but a few forecasts show a reduction in \gls{rmse} compared to the base forecasts.}
    \label{fig:rrmse_diff_sunburst_all_regions_regularisation}
\end{figure}
\begin{figure}
    \centering
    \includegraphics{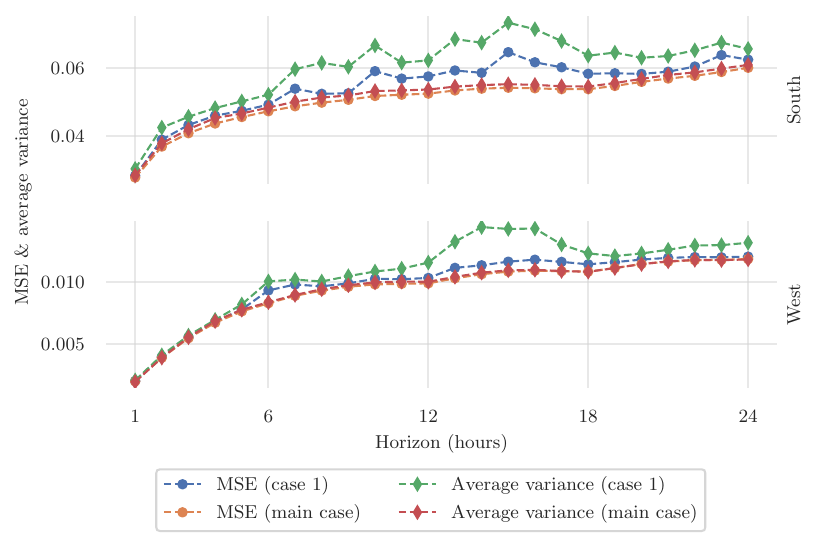} 
    \caption{MSE and averaged forecasts of prediction error variance estimates for the south and west regions based on reconciled forecasts using parameter updates every hour (main case) or every 24 hours (case 1).}
    \label{fig:mse_vs_avg_variance_other_regions}
\end{figure}

\end{document}